\documentclass[10pt, a4paper]{article}
\usepackage{etex}
\usepackage[english]{babel}
\usepackage{mathrsfs}
\usepackage{proof}
\usepackage{MnSymbol}
\usepackage{hyperref}
\usepackage{stmaryrd}
\usepackage{amsthm}
\usepackage{multirow}
\usepackage{amscd}

\theoremstyle{definition}
\newtheorem{thm}{Theorem}[section]
\theoremstyle{definition}
\newtheorem{cor}[thm]{Corrollary}
\theoremstyle{definition}
\newtheorem{vb}[thm]{Example}
\theoremstyle{definition}
\newtheorem{lem}[thm]{Lemma}
\theoremstyle{definition}
\newtheorem{rem}[thm]{Remark}
\theoremstyle{definition}
\newtheorem{defs}[thm]{Definition}

\newcommand{\valp}[1]{{v^+(#1)}}
\newcommand{\valn}[1]{{v^-(#1)}}
\newcommand{\uvalp}[1]{{v^+(#1)}}
\newcommand{\uvaln}[1]{{v^-(#1)}}

\newcommand{\inp}{+}
\newcommand{\outp}{-}

\newcommand{\semt}[1]{\mathop{\updownarrow}(#1)}

\newcommand{\ntop}{\mathop{\downarrow}}
\newcommand{\pton}{\mathop{\uparrow}}

\newcommand{\calclg}{\textbf{LG}}
\newcommand{\calccnl}{\textbf{CNL}}
\newcommand{\cnlpol}{$\textbf{CNL}^{\it pol}$}
\newcommand{\lgpol}{$\textbf{LG}^{\it pol}$}

\newcommand{\forgetp}[1]{(#1)^{\sharp}} 
\newcommand{\forgetn}[1]{(#1)^{\flat}} 
\newcommand{\forgets}[1]{(#1)^{\natural}}
\newcommand{\displ}[2]{#1\mathbin{\div}#2} 

\newcommand{\svara}{\Gamma}
\newcommand{\svarb}{\Delta}
\newcommand{\svarc}{\Theta}

\newcommand{\svard}{\Pi}
\newcommand{\svare}{\Sigma}
\newcommand{\svarf}{\Upsilon}

\newcommand{\cnlos}{\stimes{}{},\sos{}{}}
\newcommand{\cnlobs}{\stimes{}{},\sobs{}{}}
\newcommand{\grisha}[2]{A_{#1}^{#2}}
\newcommand{\grishc}[2]{C_{#1}^{#2}}

\newcommand{\lgalg}[2]{#1\leq #2}
\newcommand{\pres}[2]{#1,#2}
\newcommand{\cnlgseq}[2]{#1,#2\vdash}
\newcommand{\cnlgseqp}[2]{#1\vdash #2}

\newcommand{\presa}{\omega}

\newcommand{\defeq}{=_{\textit{def}}} 

\newcommand{\stimes}[2]{#1\bullet #2}
\newcommand{\sos}[2]{#1\leftfilledspoon #2}
\newcommand{\sobs}[2]{#1\rightfilledspoon #2}

\newcommand{\bs}{\backslash}
\newcommand{\os}{\varoslash}
\newcommand{\obs}{\varobslash}
\newcommand{\lgtimes}{\varotimes}
\newcommand{\lgplus}{\varoplus}

\newcommand{\invp}[1]{{{\Updownline} #1{\Updownline}}}

\newcommand{\suba}[1]{\sigma(#1)}
\newcommand{\subb}[1]{\tau(#1)}

\title{Focalization and phase models for classical extensions of non-associative Lambek calculus}
\author{Arno Bastenhof}

\begin{document}
\bibliographystyle{plain}
\maketitle

\begin{abstract}
Lambek's non-associative syntactic calculus (\textbf{NL}, \cite{lambek61}) excels in its resource consciousness: the usual structural rules for weakening, contraction, exchange and even associativity are all dropped. Recently, there have been proposals for conservative extensions dispensing with \textbf{NL}'s intuitionistic bias towards sequents with single conclusions: De Groote and Lamarche's \textit{classical} \textbf{NL} (\textbf{CNL}, \cite{degrootelamarche}) and Moortgat's \textit{Lambek-Grishin calculus} (\textbf{LG}, \cite{moortgat09}). We demonstrate Andreoli's focalization property (\cite{andreoli92}) for said proposals: a normalization result for Cut-free sequent derivations identifying to a large extent those differing only by trivial rule permutations. In doing so, we proceed from a `uniform' sequent presentation, deriving \textbf{CNL} from \textbf{LG} through the addition of structural rules. The normalization proof proceeds by the construction of syntactic phase models wherein every `truth' has a focused proof, similar to \cite{okada02}. 
\end{abstract}

\section{Introduction}
Logics without structural rules were first proposed by Lambek in the late fifties and early sixties. His \textit{syntactic calculus} from \cite{lambek58} was a logic of strings, making no appeal to weakening, contraction or exchange. Associativity was subsequently dropped in \cite{lambek61}, allowing for reasoning with (binary-branching) trees. These days, said calculi are referred to by the associative and non-associative \textit{Lambek calculus} respectively (\textbf{L}/\textbf{NL}).

Attempts at lifting (\textbf{N})\textbf{L}'s intuitionistic bias towards sequents with single conclusions first culminated in \textit{bilinear logic} and \textit{cyclic linear logic} (\cite{abrusci02}), both conservative extensions of \textbf{L} (hence associative). More recently, two proposals were put forward within the non-associative setting: De Groote and Lamarche's \textit{classical} \textbf{NL} (\textbf{CNL}, \cite{degrootelamarche}) and Moortgat's \textit{Lambek-Grishin calculus} (\textbf{LG}, \cite{moortgat09}). While the latter extends \textbf{NL}'s logical vocabulary by a coresiduated family of connectives, consisting of a par and coimplications, the former instead takes as primitives the tensor and par, augmented by classical linear negation, with the (co)implications reduced to defined operations.

Proof-theoretic investigations into \textbf{CNL} and \textbf{LG} have thus far concentrated on Cut-free sequent calculi and proof nets (\cite{moot07}). The current work contributes a proof of the \textit{focalization} property. First observed by Andreoli within full linear logic (\cite{andreoli92}), the latter is a normalization result for Cut-free sequent derivations, identifying (to a large extent) those that differ only by trivial rule permutations. Roughly, a focused derivation is one where invertible (logical) inferences are always applied as soon as possible, and, for every application of a non-invertible inference, the active formulas appearing in its premises are principal.

Our method of proof avoids the (local) rewriting of derivations found in the usual syntactic demonstrations of Cut elimination (\cite{laurent04}), proceeding instead by model-theoretic means similarly to \cite{okada02} and \cite{herbelinlee09}. More specifically, we propose a phase semantics (\cite{okada02}) for \textbf{CNL} and \textbf{LG}, seeking to define syntactic models wherein every `truth' has a focused derivation. The desired result then follows by composition with soundness of unfocused provability. Like in \cite{okada02}, our treatment is uniform in the sense that we demonstrate focalization of both \textbf{CNL} and \textbf{LG} by a single proof, applying as well to any further extensions by structural rules.

We proceed as follows. $\S$2 recapitulates material on \textbf{CNL} and \textbf{LG}, its modest contribution being a uniform one-sided sequent presentation, expressing their differences by a number of structural rules. The definition of focused derivations is taken up in $\S$3. As an intermediate step, we introduce \textit{polarized} adaptations of \textbf{CNL} and \textbf{LG}, recording alternations between chains of invertible and non-invertible inferences within the logical vocabulary through Girard's \textit{shift} connectives (\cite{girard00}). Phase spaces are defined in $\S$4, together with proofs of soundness and completeness. The former result is stated for polarized sequent derivations with Cut, whereas the latter concerns focused provability, thus obtaining normalization by composition. Figure \ref{summ_results} summarizes our results.\footnote{Throughout this article, in referring to a previously stated definition (lemma, theorem, corollary, figure) $n$, we often use the abbreviation D.$n$ (L.$n$, T.$n$, C.$n$, F.$n$).}

\begin{figure}
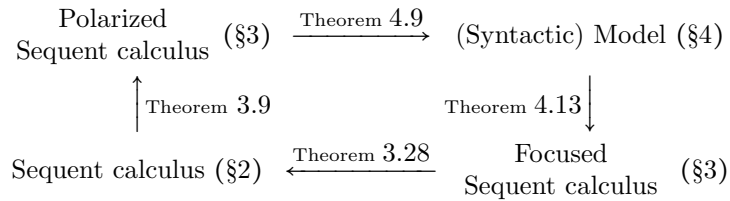


$$ \begin{CD}
\begin{array}{c}\textrm{Polarized} \\ \textrm{Sequent calculus}\end{array} (\S 3) @>\textrm{Theorem } \ref{soundness_sem}>> \textrm{(Syntactic) Model } (\S 4) \\
@AA\textrm{Theorem} \ \ref{unpol_weak}A  @V\textrm{Theorem } \ref{completeness_sem}VV \\
\textrm{Sequent calculus} \ (\S 2) @<\textrm{Theorem } \ref{strong_unpol}<< \begin{array}{c}\textrm{Focused} \\ \textrm{Sequent calculus}\end{array} \ (\S 3)
\end{CD} $$

\caption{Summary of results: composition of arrows yields focalization.}
\label{summ_results}
\end{figure}

\section{Symmetry in non-associativity}
The current section serves as a discussion of sequent presentations for \textbf{LG} and \textbf{CNL}. To simplify matters in the sequel, we aim at a `unified' treatment. Thus, we proceed from a shared logical vocabulary, detailed in $\S$2.1. 
$\S$2.2 first treats a sequent presentation for \textbf{LG}, which is essentially a one-sided play on Moortgat's display calculus (\cite{moortgat09}). By adding structural rules, we obtain in $\S$2.3 a presentation of \textbf{CNL} equivalent (provability-wise) with that of \cite{degrootelamarche}.

\subsection{Logical vocabulary}
Formulas are defined the same for \textbf{LG} and \textbf{CNL}, although some will turn out interderivable in the latter case. 
\begin{defs} Connectives come in dual pairs: a non-commutative, non-associative multiplicative disjunction (\textit{par}) accompanies a similarly resource-sensitive multiplicative conjunction (\textit{tensor}), while direction-sensitive divisions (\textit{implications}) are complemented by left- and right subtractions (\textit{coimplications}). Expanding upon \cite{moortgat09} and \cite{degrootelamarche}, we also take additives into consideration:
\begin{center}
\begin{tabular}{rcll}
$A,B$ & $::=$ & $p$ $|$ $\bar{p}$ & (Positive vs. negative atoms) \\
 & $|$ & $(A\lgtimes B)$ $|$ $(B\lgplus A)$ & (Tensor vs. par) \\
 & $|$ & $(A/B)$ $|$ $(B\obs A)$ & (Right implication vs. left coimplication) \\
 & $|$ & $(B\bs A)$ $|$ $(A\os B)$ & (Left implication vs. right coimplication) \\
 & $|$ & $(A\land B)$ $|$ $(A\lor B)$ & (Additive conjunction and disjunction)
\end{tabular}
\end{center}
\end{defs}
\noindent Note that, for each atomic formula $p$, we also assume to have at our disposal its negation $\bar{p}$. While conflicting with Moortgat's account of \textbf{LG}, 
we will find the choice of (positive/negative) polarity for atoms to influence the shape of focused proofs found in $\S$3 (cf. Example \ref{cnlg_strongfoc_exa}).

\begin{defs} Made explicit, the duality present in the above discussion is realized as a classical linear negation ($\cdot^{\bot}$):
\begin{center}
\begin{tabular}{rclcrcl}
$p^{\bot}$ & $\defeq$ & $\bar{p}$ & & $\bar{p}^{\bot}$ & $\defeq$ & $p$ \\
$(A\lgtimes B)^{\bot}$ & $\defeq$ & $B^{\bot}\lgplus A^{\bot}$ & & $(A\lgplus B)^{\bot}$ & $\defeq$ & $B^{\bot}\lgtimes A^{\bot}$ \\
$(A/B)^{\bot}$ & $\defeq$ & $B^{\bot}\obs A^{\bot}$ & & $(B\obs A)^{\bot}$ & $\defeq$ & $A^{\bot}/B^{\bot}$ \\
$(B\bs A)^{\bot}$ & $\defeq$ & $A^{\bot}\os B^{\bot}$ & & $(A\os B)^{\bot}$ & $\defeq$ & $B^{\bot}\bs A^{\bot}$ \\
$(A\land B)^{\bot}$ & $\defeq$ & $B^{\bot}\lor A^{\bot}$ & & $(A\lor B)^{\bot}$ & $\defeq$ & $B^{\bot}\land A^{\bot}$
\end{tabular}
\end{center}
\end{defs}
\noindent Indeed, involutivity ($A^{\bot\bot}=A$) is easily established. Note that we have not identified $A/B$ with $A\lgplus B^{\bot}$, and similarly for the other (co)implications. Instead, said formulas will turn out interderivable in \textbf{CNL}, though not in \textbf{LG}.

\subsection{\textbf{LG} sequentialized}
Our presentation of \textbf{LG} proceeds in two steps. First, in $\S$2.2.1, we briefly recapitulate the algebraic account, adapted to the extended logical vocabulary. $\S$2.2.2 introduces our sequent calculus and justifies its rules by dual translations into algebraic derivations.

\subsubsection{The minimal logic of (co)residuation}
Defined algebraically, derivability in \textbf{LG} characterizes inequalities $\lgalg{A}{B}$. Figure \ref{lg_algebraic} presents an axiomatization: $\leq$ satisfies the preorder laws I(dentity) and T(ransitivity), $\land$/$\lor$ are realized as meets/joins, $/,\bs$ act as residuals to $\lgtimes$ $(r)$, while, finally, the connectives $\lgplus$, $\os$ and $\obs$ listen to the dual \textit{coresiduation} laws $(cr)$, obtained by reversing $\leq$.
Note the complete absence of axioms licensing the structural rules of sequent calculi: the tensor and par fail to satisfy weakening (e.g., $\lgalg{A\lgtimes B}{A}$ or $\lgalg{A}{A\lgplus B}$) and contraction ($\lgalg{A}{A\lgtimes A}$ or $\lgalg{A\lgplus A}{A}$), nor are they associative or commutative.

\begin{figure}
\begin{center}
\begin{tabular}{c}
Preorder laws \\ \\
\begin{tabular}{ccc}
$\infer[I]{\lgalg{A}{A}}{}$ & & $\infer[T]{\lgalg{A}{C}}{\lgalg{A}{B} & \lgalg{B}{C}}$ 
\end{tabular} \\ \\
(Co)residuation \\ \\
\begin{tabular}{ccccccc}
$\infer=[r]{\lgalg{A}{C/B}}{\lgalg{A\lgtimes B}{C}}$ & & $\infer=[r]{\lgalg{B}{A\bs C}}{\lgalg{A\lgtimes B}{C}}$ & &
$\infer=[cr]{\lgalg{A\obs C}{B}}{\lgalg{C}{A\lgplus B}}$ & & $\infer=[cr]{\lgalg{C\os B}{A}}{\lgalg{C}{A\lgplus B}}$
\end{tabular} \\ \\
Meets/Joins \\ \\
\begin{tabular}{ccccc}
$\infer=[\land]{\lgalg{A}{B\land C}}{\lgalg{A}{B} & \lgalg{A}{C}}$ & & 
$\infer=[\lor]{\lgalg{A\lor B}{C}}{\lgalg{A}{C} & \lgalg{B}{C}}$
\end{tabular}
\end{tabular}
\end{center}
\caption{\textbf{LG} characterized algebraically. Double horizontal inference lines indicate interchangeability of premises and conclusion, where multiple conclusions (in the rules for meets/joins) are to be interpreted conjunctively.}
\label{lg_algebraic}
\end{figure}
\begin{figure}[h]
\begin{center}
\begin{tabular}{ccc}
$\infer[\grisha{I}{1}]{\lgalg{C\lgtimes A}{D\lgplus B}}{\lgalg{A\os B}{C\bs D}}$ & &
$\infer[\grisha{I}{2}]{\lgalg{B\lgtimes D}{A\lgplus C}}{\lgalg{A\obs B}{C/D}}$ \\ \\

$\infer[\grishc{I}{1}]{\lgalg{C\lgtimes B}{A\lgplus D}}{\lgalg{A\obs B}{C\bs D}}$ & &
$\infer[\grishc{I}{2}]{\lgalg{A\lgtimes D}{C\lgplus B}}{\lgalg{A\os B}{C/D}}$ 
%
%
\end{tabular}
\end{center}
\caption{Linear distributivity laws for the Lambek-Grishin calculus.}
\label{lgalg_struc}
\end{figure}

\begin{lem}\label{derived_mon}
Derived rules of inference include the following monotonicity laws:
\begin{center}
\begin{tabular}{ccc}
$\infer[m]{\begin{array}{c}\lgalg{A\lgtimes C}{B\lgtimes D} \\ \lgalg{A/D}{B/C} \\ \lgalg{D\bs A}{C\bs B}\end{array}}{\lgalg{A}{B} & \lgalg{C}{D}}$ & \ \ &
$\infer[m]{\begin{array}{c}\lgalg{A\lgplus C}{B\lgplus D} \\ \lgalg{A\os D}{B\os C} \\ \lgalg{D\obs A}{C\obs B}\end{array}}{\lgalg{A}{B} & \lgalg{C}{D}}$
\end{tabular}
\end{center}
\end{lem}
\begin{proof}
As a typical case, we take the rule for $\lgtimes$, derived thus:
\begin{center}
$\infer[r]{\lgalg{A\lgtimes C}{B\lgtimes D}}{
\infer[T]{\lgalg{A}{(B\lgtimes D)/C}}{
\lgalg{A}{B} &
\infer[r]{\lgalg{B}{(B\lgtimes D)/C}}{
\infer[r]{\lgalg{B\lgtimes C}{B\lgtimes D}}{
\infer[T]{\lgalg{C}{B\bs(B\lgtimes D)}}{
\lgalg{C}{D} &
\infer[r]{\lgalg{D}{B\bs(B\lgtimes D)}}{
\infer[I]{\lgalg{B\lgtimes D}{B\lgtimes D}}{}}}}}}}$
\end{center}
\end{proof}
\noindent Grishin (\cite{grishin}) has considered possible extensions of \textbf{LG} by groups of axioms establishing interaction between $\{\lgtimes,/,\bs\}$ and $\{\lgplus,\obs,\os\}$, while remaining conservative over the two families separately. To illustrate the applicability of our method to Grishin's studies, we single out the laws for linear distributivity of $\lgtimes$ over $\lgplus$ (\cite{cockettseely}), presented in rule format in Figure \ref{lgalg_struc}.

\begin{lem}
The following inequalities are derivable in the presence of the rules of Figure \ref{lgalg_struc}.
\begin{center}
\begin{tabular}{lcl}
$\lgalg{A\lgtimes(B\lgplus C)}{(A\lgtimes B)\lgplus C}$ & &
$\lgalg{(A\lgplus B)\lgtimes C}{A\lgplus(B\lgtimes C)}$ \\
$\lgalg{A\lgtimes(B\lgplus C)}{B\lgplus(A\lgtimes C)}$ & &
$\lgalg{(A\lgplus B)\lgtimes C}{(A\lgtimes C)\lgplus B}$
\end{tabular}
\end{center}
\end{lem}
\begin{proof}
As a typical case, we check $\lgalg{A\lgtimes(B\lgplus C)}{(A\lgtimes B)\lgplus C}$.
\begin{center}
$\infer[\grisha{I}{1}]{\lgalg{A\lgtimes(B\lgplus C)}{(A\lgtimes B)\lgplus C}}{
\infer[T]{\lgalg{(B\lgplus C)\os C}{A\bs(A\lgtimes B)}}{
\infer[cr]{\lgalg{(B\lgplus C)\os C}{B}}{\infer[I]{\lgalg{B\lgplus C}{B\lgplus C}}{}} &
\infer[r]{\lgalg{B}{A\bs(A\lgtimes B)}}{\infer[I]{\lgalg{A\lgtimes B}{A\lgtimes B}}{}}}}$
\end{center}
\end{proof}
\noindent As stressed by Moortgat, one may argue there to be different conceptions of \calclg, depending on which of Grishin's groups are adopted. Therefore, we will henceforth refer by (the algebraic presentation of) \calclg$_{\emptyset}$ to the calculus defined by the inference rules of Figure \ref{lg_algebraic} only, while \calclg$_I$ denotes the extension by linear distributivity (following Moortgat's notation). On those occassions where the difference is inessential, we keep using \calclg. For a more thorough exploration of the wider landscape of Lambek-Grishin calculi, the reader is referred to \cite{moortgat09}. We note that the methods used in this article are general enough so as to be applicable to said alternatives. We conclude with the realization of $\cdot^{\bot}$ at the level of derivations, again demonstrable through a straightforward induction:
\begin{lem}
For any $A,B$, $\lgalg{A}{B}$ iff $\lgalg{B^{\bot}}{A^{\bot}}$.
\end{lem}

\subsubsection{One-sided sequents}
As shown by Moortgat (\cite{moortgat09}), \textbf{LG} has a Cut-free display calculus. Like in ordinary two-sided calculi, connectives are introduced as hypotheses or conclusions. To guarantee Cut-admissibility, however, structural commas no longer associate exclusively to (multiplicative) conjunctions and disjunctions, but may also appear as counterparts for the (co)implications. The current section presents a one-sided retelling of Moortgat's display calculus, in the sense that, for any $A,B$, the inequalities $\lgalg{A}{B}$ and $\lgalg{B^{\bot}}{A^{\bot}}$ will have the same sequent counterpart.

\begin{defs}\label{structures_displ}
Proofs establish
\textit{presentations}, being pairs of \textit{structures} $\Gamma,\Delta$:
\begin{center}
\begin{tabular}{rclr}
$\Gamma,\Delta,\Theta$ & $::=$ & $A$ $|$ $(\stimes{\Gamma}{\Delta})$ $|$ $(\sos{\Gamma}{\Delta})$ $|$ $(\sobs{\Delta}{\Gamma})$ & \ \ Structures \\
$\presa$ & $::=$ & $\pres{\Gamma}{\Delta}$ & \ \ Presentations
\end{tabular}
\end{center}
\end{defs}
\noindent Terminology is adapted from Andreoli (\cite{andreoli04}), who distinguished between \textit{presentations} and \textit{varieties}. The intuition, further pursued below, is that presentations are closed under reversible structural rules allowing any of its formulas to be displayed as the whole of one of its components. The equivalence classes of presentations generated by said rules are the (freely generated) varieties, presentations thus `presenting' a variety from the point of view of one of its substructures (particularly formulas). We refrain from explicating the latter concept, however, as the focused derivations defined in $\S$3.3 already compile away the reversible structural rules. We note Lamarche (\cite{lamarche03}) uses `terms' and `reversible terms' for denoting syntactic objects similar to our structures and Andreoli's varieties.

\begin{defs} Structures $\Gamma$ interpret by pairs of dual formulas $\Gamma^{\inp}$, $\Gamma^{\outp}$:
\begin{center}
\begin{tabular}{rclcrcl}
$A^{\inp}$ & $\defeq$ & $A$ & & $A^{\outp}$ & $\defeq$ & $A^{\bot}$ \\
$(\stimes{\Gamma}{\Delta})^{\inp}$ & $\defeq$ & $\Gamma^{\inp}\lgtimes\Delta^{\inp}$ & & $(\stimes{\Gamma}{\Delta})^{\outp}$ & $\defeq$ & $\Delta^{\outp}\lgplus\Gamma^{\outp}$ \\
$(\sobs{\Delta}{\Gamma})^{\inp}$ & $\defeq$ & $\Delta^{\outp}\obs\Gamma^{\inp}$ & & $(\sobs{\Delta}{\Gamma})^{\outp}$ & $\defeq$ & $\Gamma^{\outp}/\Delta^{\inp}$ \\
$(\sos{\Gamma}{\Delta})^{\inp}$ & $\defeq$ & $\Gamma^{\inp}\os\Delta^{\outp}$ & & $(\sos{\Gamma}{\Delta})^{\inp}$ & $\defeq$ & $\Delta^{\inp}\bs\Gamma^{\outp}$
\end{tabular}
\end{center}
\end{defs}
\noindent One easily shows $\Gamma^{\inp}=\Gamma^{\outp\bot}$ and $\Gamma^{\outp}=\Gamma^{\inp\bot}$. The ambiguity extends to the level of derivability, where we will find a witness for the provability of $\pres{\Gamma}{\Delta}$ to be realizable into algebraic derivations of both $\lgalg{\Gamma^{\inp}}{\Delta^{\outp}}$ and $\lgalg{\Delta^{\inp}}{\Gamma^{\outp}}$. Conversely, inequalities $\lgalg{A}{B}$ and $\lgalg{B^{\bot}}{A^{\bot}}$ are both presented by $\pres{A}{B^{\bot}}$.

\begin{defs}
Figure \ref{neutraldispl} defines derivability judgements $\presa\vdash$ for \textbf{LG}, written as left-sided sequents.\footnote{Our (non-conventional) preference for left-sided sequents over right-sided ones is motivated by the former's transparent correspondence with intuitionistic sequents, constituting the target of double negation translations (\cite{moortgat09}).} Next to the familiar axioms and Cut, we have the \textit{display postulates} $(dp)$, ensuring, for any presentation $\presa$ and an occurrence therein of a formula $A$, the (unique) existence of $\svarb$ s.t. $\presa$ may be rewritten into $\presa'=\pres{\svarb}{A}$. We say $A$ is \textit{displayed} in $\presa'$. Presentations $\presa$ and $\presa'$ interderivable through display postulates are said to be \textit{display equivalent}, a situation often abbreviated
\begin{center}
$\infer=[Dp]{\presa\vdash}{\presa'\vdash}$
\end{center}
\end{defs}

\begin{figure}
\begin{center}
\begin{tabular}{ccc}
 & $\infer[Ax]{\cnlgseq{A}{A^{\bot}}}{}$ &
$\infer[Cut]{\cnlgseq{\Gamma}{\Delta}}{\cnlgseq{\Delta}{A} & \cnlgseq{\Gamma}{A^{\bot}}}$ \\ \\
$\infer[dp]{\cnlgseq{\Gamma}{\Delta}}{\cnlgseq{\Delta}{\Gamma}}$ &
$\infer=[dp]{\cnlgseq{\Gamma}{\sobs{\Delta}{\Theta}}}{\cnlgseq{\stimes{\Gamma}{\Delta}}{\Theta}}$ &
$\infer=[dp]{\cnlgseq{\sos{\Gamma}{\Delta}}{\Theta}}{\cnlgseq{\Gamma}{\stimes{\Delta}{\Theta}}}$ \\ \\
$\infer[\lgtimes]{\cnlgseq{\Gamma}{A\lgtimes B}}{\cnlgseq{\Gamma}{\stimes{A}{B}}}$ &
$\infer[\obs]{\cnlgseq{\Gamma}{B\obs A}}{\cnlgseq{\Gamma}{\sobs{B^{\bot}}{A}}}$ &
$\infer[\os]{\cnlgseq{\Gamma}{A\os B}}{\cnlgseq{\Gamma}{\sos{A}{B^{\bot}}}}$ \\ \\
$\infer[\lgplus]{\cnlgseq{\stimes{\Delta}{\Gamma}}{A\lgplus B}}{\cnlgseq{\Gamma}{A} & \cnlgseq{\Delta}{B}}$ &
$\infer[/]{\cnlgseq{\sobs{\Delta}{\Gamma}}{A/B}}{\cnlgseq{\Delta}{B^{\bot}} & \cnlgseq{\Gamma}{A}}$ &
$\infer[\bs]{\cnlgseq{\sos{\Gamma}{\Delta}}{B\bs A}}{\cnlgseq{\Delta}{B^{\bot}} & \cnlgseq{\Gamma}{A}}$ \\ \\
$\infer[\land^l]{\cnlgseq{\Gamma}{A\land B}}{\cnlgseq{\Gamma}{A}}$ &
$\infer[\land^r]{\cnlgseq{\Gamma}{A\land B}}{\cnlgseq{\Gamma}{B}}$ &
$\infer[\lor]{\cnlgseq{\Gamma}{A\lor B}}{\cnlgseq{\Gamma}{A} & \cnlgseq{\Gamma}{B}}$
\end{tabular}
\end{center}
\caption{A left-sided sequent calculus for \textbf{LG}: base logic}
\label{neutraldispl}
\end{figure}

\begin{figure}
\begin{center}
\begin{tabular}{ccc}
$\infer[\grisha{I}{1}]{\cnlgseq{\stimes{\Gamma_1}{\Gamma_2}}{\stimes{\Delta_2}{\Delta_1}}}{
\cnlgseq{\sos{\Gamma_2}{\Delta_2}}{\sos{\Delta_1}{\Gamma_1}}}$ &
$\infer[\grisha{I}{2}]{\cnlgseq{\stimes{\Gamma_1}{\Gamma_2}}{\stimes{\Delta_2}{\Delta_1}}}{
\cnlgseq{\sobs{\Delta_1}{\Gamma_1}}{\sobs{\Gamma_2}{\Delta_2}}}$ &
$\infer[\grishc{I}{}]{\cnlgseq{\stimes{\Gamma_1}{\Gamma_2}}{\stimes{\Delta_2}{\Delta_1}}}{
\cnlgseq{\sos{\Delta_2}{\Gamma_1}}{\sobs{\Delta_1}{\Gamma_2}}}$

\end{tabular}
\end{center}
\caption{Structural rules: Linear Distributivity}
\label{neutraldispl_grishin}
\end{figure}

\begin{defs}
We fix terminology for referring to occurrences of formulas in (instances of) logical rules. Given one of the form
\begin{center}
$\infer[R]{\cnlgseq{\svara}{A}}{\cnlgseq{\svara_1}{A_1} & \dots & \cnlgseq{\svara_n}{A_n}}$
\end{center}
we call $A$ the \textit{main} or \textit{principal} formula of $R$, and the subformulas $A_1,\dots,A_n$ of $A$ occurring in the premises the \textit{active} formulas of $R$. 
We also say $A_i$ ($1\leq i\leq n$) is principal if main in the rule deriving the corresponding premise.
\end{defs}

\begin{vb}\label{example_deriv_seqlg}
Figure \ref{example_deriv} witnesses $\cnlgseq{\stimes{(\stimes{p/q}{q})}{p\bs r}}{\bar{r}}$ and $\cnlgseq{\stimes{p/(q\bs p)}{(p/(q\bs p))\bs p}}{\bar{p}}$ by two derivations each, only one of which will be preserved by focused proof search in the former case. In contrast, the other two derivations employ different axiom matchings, in the precise sense that they generalize to distinct presentations $\pres{\stimes{a/(d\bs b)}{(c/(d\bs c))\bs b}}{\bar{a}}$ and $\pres{\stimes{a/(c\bs b)}{(a/(c\bs b))\bs d}}{\bar{d}}$, and therefore remain distinct under focalization.\footnote{In general, any two derivations of the same sequent employing different axiom matchings will likewise remain distinct under focusing, although the converse need not always hold. In particular, focalization still identifies less derivations than proof nets do.}
\end{vb}

\noindent We proceed to show soundness and completeness w.r.t. algebraic derivability.

\begin{figure}
\begin{center}
\begin{tabular}{c}
%

\begin{tabular}{c}$\infer=[Dp]{\cnlgseq{\stimes{p/(q\bs p)}{(p/(q\bs p))\bs p}}{\bar{p}}}{
\infer[/]{\cnlgseq{\sobs{(p/(q\bs p))\bs p}{\bar{p}}}{p/(q\bs p)}}{
\infer[\os]{\cnlgseq{(p/(q\bs p))\bs p}{\bar{p}\os\bar{q}}}{
\infer=[Dp]{\cnlgseq{(p/(q\bs p))\bs p}{\sos{\bar{p}}{q}}}{
\infer[\bs]{\cnlgseq{\sos{\bar{p}}{q}}{(p/(q\bs p))\bs p}}{
\infer=[Dp]{\cnlgseq{(\bar{p}\os\bar{q})\obs\bar{p}}{q}}{
\infer[\obs]{\cnlgseq{q}{(\bar{p}\os\bar{q})\obs\bar{p}}}{
\infer=[Dp]{\cnlgseq{q}{\sobs{q\bs p}{\bar{p}}}}{
\infer[\bs]{\cnlgseq{\sos{\bar{p}}{q}}{q\bs p}}{
\infer[I]{\cnlgseq{q}{\bar{q}}}{} &
\infer[I]{\cnlgseq{\bar{p}}{p}}{}}}}} &
\infer[I]{\cnlgseq{\bar{p}}{p}}{}}}} &
\infer[I]{\cnlgseq{\bar{p}}{p}}{}}}$ \\ \\

$\infer=[Dp]{\cnlgseq{\stimes{p/(q\bs p)}{(p/(q\bs p))\bs p}}{\bar{p}}}{
\infer[\bs]{\cnlgseq{\sos{\bar{p}}{p/(q\bs p)}}{(p/(q\bs p))\bs p}}{
\infer[I]{\cnlgseq{p/(q\bs p)}{(\bar{p}\os\bar{q})\obs\bar{p}}}{} &
\infer[I]{\cnlgseq{\bar{p}}{p}}{}}}$
\end{tabular}

\begin{tabular}{ccc}
$\infer=[Dp]{\cnlgseq{\stimes{(\stimes{p/q}{q})}{p\bs r}}{\bar{r}}}{
\infer[/]{\cnlgseq{\sobs{q}{(\sobs{p\bs r}{\bar{r}})}}{p/q}}{
\infer[I]{\cnlgseq{q}{\bar{q}}}{} &
\infer=[Dp]{\cnlgseq{\sobs{p\bs r}{\bar{r}}}{p}}{
\infer[\bs]{\cnlgseq{\sos{\bar{r}}{p}}{p\bs r}}{
\infer[I]{\cnlgseq{p}{\bar{p}}}{} &
\infer[I]{\cnlgseq{\bar{r}}{r}}{}}}}}$ \\ \\ \\ \\

$\infer=[Dp]{\cnlgseq{\stimes{(\stimes{p/q}{q})}{p\bs r}}{\bar{r}}}{
\infer[\bs]{\cnlgseq{\sos{\bar{r}}{(\stimes{p/q}{q})}}{p\bs r}}{
\infer=[Dp]{\cnlgseq{\stimes{p/q}{q}}{\bar{p}}}{
\infer[/]{\cnlgseq{\sobs{q}{\bar{p}}}{p/q}}{
\infer[I]{\cnlgseq{q}{\bar{q}}}{} &
\infer[I]{\cnlgseq{\bar{p}}{p}}{}}} &
\infer[I]{\cnlgseq{\bar{r}}{r}}{}}}$
\end{tabular}
\end{tabular}
\end{center}

\caption{Example derivations witnessing $\cnlgseq{\stimes{p/(q\bs p)}{(p/(q\bs p))\bs p}}{\bar{p}}$ and $\cnlgseq{\stimes{(\stimes{p/q}{q})}{p\bs r}}{\bar{r}}$.}
\label{example_deriv}
\end{figure}

\begin{lem}\label{completeness_alg_proof} If $\lgalg{A}{B}$, then $\cnlgseq{A}{B^{\bot}}$.
\end{lem}
\begin{proof}
By an induction on the derivation witnessing $\lgalg{A}{B}$. The preorder laws trivially translate to Axioms and Cut, while the rules for meets and joins are immediate by $(\lor)$ and $(\land^{l/r})$. This leaves us with the (co)residuation laws and the Grishin interactions as the only nontrivial cases. \\

\noindent\textbf{(Co)residuation}. We explicitly check $\lgalg{A}{C/B}$ if $\lgalg{A\lgtimes B}{C}$, the other cases being similar. By induction hypothesis, we know $\cnlgseq{A\lgtimes B}{C^{\bot}}$. Hence,
\begin{center}
$\infer[\obs]{\cnlgseq{A}{B^{\bot}\obs C^{\bot}}}{
\infer=[Dp]{\cnlgseq{A}{\sobs{B}{C^{\bot}}}}{
\infer[T]{\cnlgseq{\stimes{A}{B}}{C^{\bot}}}{
\infer[\lgplus]{\cnlgseq{\stimes{A}{B}}{B^{\bot}\lgplus A^{\bot}}}{
\infer[I]{\cnlgseq{B}{B^{\bot}}}{} &
\infer[I]{\cnlgseq{A}{A^{\bot}}}{}} &
\infer=[Dp]{\cnlgseq{C^{\bot}}{A\lgtimes B}}{
\cnlgseq{A\lgtimes B}{C^{\bot}}}}}}$
\end{center}

\noindent\textbf{Grishin interactions}. Again, we consider only one case. 
Take $\lgalg{C\lgtimes A}{D\lgtimes B}$ if $\lgalg{A\os B}{C\bs D}$. By induction hypothesis, $\cnlgseq{A\os B}{D^{\bot}\os C^{\bot}}$. Hence,
\begin{center}
$\infer[\lgtimes]{\cnlgseq{C\lgtimes A}{B^{\bot}\lgtimes D^{\bot}}}{
\infer=[Dp]{\cnlgseq{C\lgtimes A}{\stimes{B^{\bot}}{D^{\bot}}}}{
\infer[\lgtimes]{\cnlgseq{\stimes{B^{\bot}}{D^{\bot}}}{C\lgtimes A}}{
\infer[\grisha{I}{1}]{\cnlgseq{\stimes{B^{\bot}}{D^{\bot}}}{\stimes{C}{A}}}{
\infer[\os]{\cnlgseq{\sos{D^{\bot}}{C}}{\sos{A}{B^{\bot}}}}{
\infer[T]{\cnlgseq{\sos{D^{\bot}}{C}}{A\os B}}{
\cnlgseq{A\os B}{D^{\bot}\os C^{\bot}} &
\infer[\bs]{\cnlgseq{\sos{D^{\bot}}{C}}{C\bs D}}{
\infer[I]{\cnlgseq{C}{C^{\bot}}}{} &
\infer[I]{\cnlgseq{D^{\bot}}{D}}{}}}}}}}}$
\end{center}
\end{proof}

\begin{lem}\label{soundness_alg_proof}
If $\cnlgseq{\Gamma}{\Delta}$, then $\lgalg{\Gamma^{\inp}}{\Delta^{\outp}}$ and $\lgalg{\Delta^{\inp}}{\Gamma^{\outp}}$.
\end{lem}
\begin{proof}
We proceed by induction on $\cnlgseq{\Gamma}{\Delta}$. The cases \textit{(Ax)} and \textit{(Cut)} are trivial, while $(\lgtimes)$, $(\os)$ and $(\obs)$ are immediate from the induction hypotheses. \\ 

\noindent\textbf{Case} $(dp)$. In general, the display postulates are justified by (co)residuation. As a typical case, we take 
$\cnlgseq{\stimes{\Gamma}{\Delta}}{\Theta}$ implies $\cnlgseq{\Gamma}{\sobs{\Delta}{\Theta}}$. By the induction hypothesis, $\lgalg{\Gamma^{\inp}\lgtimes\Delta^{\inp}}{\Theta^{\outp}}$ and $\lgalg{\Theta^{\inp}}{\Delta^{\outp}\lgplus\Gamma^{\outp}}$. Hence,
\begin{center}
\begin{tabular}{ccc}
$\infer[r]{\lgalg{\Gamma^{\inp}}{\Theta^{\outp}/\Delta^{\inp}}}{\infer[IH]{\lgalg{\Gamma^{\inp}\lgtimes\Delta^{\inp}}{\Theta^{\outp}}}{}}$ & \ \ & 
$\infer[r]{\lgalg{\Theta^{\inp}\os\Delta^{\outp}}{\Gamma^{\outp}}}{\infer[IH]{\lgalg{\Theta^{\inp}}{\Delta^{\outp}\lgplus\Gamma^{\outp}}}{}}$
\end{tabular}
\end{center}

\noindent\textbf{Cases} $(\lgplus)$, $(/)$ \textbf{and} $(\bs)$. In general, said cases all depend on Lemma \ref{derived_mon}. Consider $(/)$. Assuming $\cnlgseq{\Delta}{B^{\bot}}$ and $\cnlgseq{\Gamma}{A}$, we have as induction hypotheses $\lgalg{B^{\bot}}{\Delta^{\outp}}$, $\lgalg{\Delta^{\inp}}{B}$, $\lgalg{A}{\Gamma^{\outp}}$ and $\lgalg{\Gamma^{\inp}}{A^{\bot}}$. Hence,
\begin{center}
\begin{tabular}{ccc}
$\infer[m]{\lgalg{\Delta^{\outp}\obs\Gamma^{\inp}}{B^{\bot}\obs A^{\bot}}}{\infer[IH]{\lgalg{\Gamma^{\inp}}{A^{\bot}}}{} & \infer[IH]{\lgalg{B^{\bot}}{\Delta^{\outp}}}{}}$ & \ \ &
$\infer[m]{\lgalg{A/B}{\Gamma^{\outp}/\Delta^{\inp}}}{\infer[IH]{\lgalg{A}{\Gamma^{\outp}}}{} & \infer[IH]{\lgalg{\Delta^{\inp}}{B}}{}}$
\end{tabular}
\end{center}

\noindent\textbf{Cases} $(\land^l)$ \textbf{and} $(\land^r)$. We consider $(\land^l)$, $(\land^r)$ being similar. Assuming $\cnlgseq{\Gamma}{A}$, we have induction hypotheses $\lgalg{\Gamma^{\inp}}{A^{\bot}}$ and $\lgalg{A}{\Gamma^{\outp}}$. Hence,
\begin{center}
\begin{tabular}{ccc}
$\infer[\lor]{\lgalg{\Gamma^{\inp}}{B^{\bot}\lor A^{\bot}}}{\infer[IH]{\lgalg{\Gamma^{\inp}}{A^{\bot}}}{}}$ & \ \ &
$\infer[\lor]{\lgalg{A\land B}{\Gamma^{\outp}}}{\infer[IH]{\lgalg{A}{\Gamma^{\outp}}}{}}$
\end{tabular}
\end{center}

\noindent\textbf{Case} $(\lor)$. Assuming $\cnlgseq{\Gamma}{A}$ and $\cnlgseq{\Gamma}{B}$, we have as induction hypotheses $\lgalg{\Gamma^{\inp}}{A^{\bot}}$, $\lgalg{A}{\Gamma^{\outp}}$, $\lgalg{\Gamma^{\inp}}{B^{\bot}}$ and $\lgalg{B}{\Gamma^{\outp}}$. Hence,
\begin{center}
\begin{tabular}{ccc}
$\infer[\land]{\lgalg{\Gamma^{\inp}}{B^{\bot}\land A^{\bot}}}{\infer[IH]{\lgalg{\Gamma^{\inp}}{A^{\bot}}}{} & \infer[IH]{\lgalg{\Gamma^{\inp}}{B^{\bot}}}{}}$ & \ \ &
$\infer[\land]{\lgalg{A\lor B}{\Gamma^{\outp}}}{\infer[IH]{\lgalg{A}{\Gamma^{\outp}}}{} & \infer[IH]{\lgalg{B}{\Gamma^{\outp}}}{}}$
\end{tabular}
\end{center}

\noindent\textbf{Case} 
$(\grisha{I}{1/2}/\grishc{I}{})$. We consider $\cnlgseq{\sos{\Delta_2}{\Gamma_1}}{\sobs{\Delta_1}{\Gamma_2}}$ if $\cnlgseq{\stimes{\Gamma_1}{\Gamma_2}}{\stimes{\Delta_2}{\Delta_1}}$ as a typical case. By induction hypothesis, $\lgalg{\Delta_1^{\outp}\obs\Gamma_2^{\inp}}{\Gamma_1^{\inp}\bs\Delta_2^{\outp}}$ and $\lgalg{\Delta_2^{\inp}\os\Gamma_1^{\outp}}{\Gamma_2^{\outp}/\Delta_1^{\inp}}$:
\begin{center}
\begin{tabular}{ccc}
\begin{tabular}{ccc}
$\infer[\grishc{I}{1}]{\lgalg{\Gamma_1^{\inp}\lgtimes\Gamma_2^{\inp}}{\Delta_1^{\outp}\lgplus\Delta_2^{\outp}}}{\infer[IH]{\lgalg{\Delta_1^{\outp}\obs\Gamma_2^{\inp}}{\Gamma_1^{\inp}\bs\Delta_2^{\outp}}}{}}$ & \ \ &
$\infer[\grishc{I}{2}]{\lgalg{\Delta_2^{\inp}\lgtimes\Delta_1^{\inp}}{\Gamma_2^{\outp}\lgplus\Gamma_1^{\outp}}}{\infer[IH]{\lgalg{\Delta_2^{\inp}\os\Gamma_1^{\outp}}{\Gamma_2^{\outp}/\Delta_1^{\inp}}}{}}$
\end{tabular}
\end{tabular}
\end{center}
\end{proof}

\subsection{CNL sequentialized}
\textbf{CNL} derives from $\textbf{LG}_{\emptyset}$ by identifying $\stimes{}{}$, $\sos{}{}$ and $\sobs{}{}$, i.e., by adding the structural rules of Figure \ref{neutraldispl_cnl}.
\begin{figure}
\begin{center}
\begin{tabular}{ccc}
$\infer=[\cnlos]{\cnlgseq{\Gamma}{\stimes{\Delta}{\Theta}}}{\cnlgseq{\Gamma}{\sos{\Delta}{\Theta}}}$ & &
$\infer=[\cnlobs]{\cnlgseq{\Gamma}{\stimes{\Delta}{\Theta}}}{\cnlgseq{\Gamma}{\sobs{\Delta}{\Theta}}}$
\end{tabular}
\end{center}
\caption{Structural rules: Classical non-associative Lambek calculus}
\label{neutraldispl_cnl}
\end{figure}
\begin{lem} In \calccnl, we have $\cnlgseq{B\bs A}{A^{\bot}\lgtimes B}$, $\cnlgseq{B\lgplus A^{\bot}}{A\os B}$, $\cnlgseq{A/B}{B\lgtimes A^{\bot}}$ and $\cnlgseq{A^{\bot}\lgplus B}{B\obs A}$.
\end{lem}
\begin{proof} We demonstrate the first two claims, the latter two being similar.
\begin{center}
\begin{tabular}{ccc}
$\infer[\lgtimes]{\cnlgseq{B\bs A}{A^{\bot}\lgtimes B}}{
\infer={\cnlgseq{B\bs A}{\stimes{A^{\bot}}{B}}}{
\infer=[Dp]{\cnlgseq{B\bs A}{\sos{A^{\bot}}{B}}}{
\infer[\bs]{\cnlgseq{\sos{A^{\bot}}{B}}{B\bs A}}{
\infer[Ax]{\cnlgseq{B}{B^{\bot}}}{} &
\infer[Ax]{\cnlgseq{A^{\bot}}{A}}{}}}}}$ & &

$\infer[\os]{\cnlgseq{B\lgplus A^{\bot}}{A\os B}}{
\infer={\cnlgseq{B\lgplus A^{\bot}}{\sos{A}{B^{\bot}}}}{
\infer=[Dp]{\cnlgseq{B\lgplus A^{\bot}}{\stimes{A}{B^{\bot}}}}{
\infer[\lgplus]{\cnlgseq{\stimes{A}{B^{\bot}}}{B\lgplus A^{\bot}}}{
\infer[Ax]{\cnlgseq{A}{A^{\bot}}}{} &
\infer[Ax]{\cnlgseq{B^{\bot}}{B}}{}}}}}$

\end{tabular}
\end{center}
\end{proof}
\noindent When mapped into inequalities $\lgalg{A}{B}$ via $\cdot^{\inp}$ and $\cdot^{\outp}$, the previous lemma suggests the following identifications between formulas:
\begin{center}
\begin{tabular}{rclcrcl}
$A/B$ & $=$ & $A\lgplus B^{\bot}$ & & $B\obs A$ & $=$ & $B^{\bot}\lgtimes A$ \\
$B\bs A$ & $=$ & $B^{\bot}\lgplus A$ & & $A\os B$ & $=$ & $A\lgtimes B^{\bot}$
\end{tabular}
\end{center}
In particular, we have the more economic axiomatization for \calccnl \ of Figure \ref{cnl_degrlamstyle}, as originally employed in \cite{degrootelamarche}, save for a few notational differences. We here stick to the presentation of \calccnl \ as derived from \calclg$_{\emptyset}$ with structural postulates, as it allows for a uniform proof of the focalization property. 
\begin{figure}
\begin{center}
\begin{tabular}{c}
\begin{tabular}{rcl}
$A,B$ & $::=$ & $p$ $|$ $\bar{p}$ $|$ $(A\lgtimes B)$ $|$ $(A\lgplus B)$ $|$ $(A\land B)$ $|$ $(A\lor B)$ \\
$\Gamma,\Delta$ & $::=$ & $A$ $|$ $(\stimes{\Gamma}{\Delta})$
\end{tabular} \\ \\
\begin{tabular}{ccc}
$\infer[Ax]{\cnlgseq{A}{A^{\bot}}}{}$ &
$\infer[dp]{\cnlgseq{\Delta}{\Gamma}}{\cnlgseq{\Gamma}{\Delta}}$ &
$\infer=[dp]{\cnlgseq{\Gamma}{\stimes{\Delta}{\Theta}}}{\cnlgseq{\stimes{\Gamma}{\Delta}}{\Theta}}$ \\ \\
$\infer[\lgtimes]{\cnlgseq{\Gamma}{A\lgtimes B}}{\cnlgseq{\Gamma}{\stimes{A}{B}}}$ &
$\infer[\lgplus]{\cnlgseq{\stimes{\Delta}{\Gamma}}{A\lgplus B}}{\cnlgseq{\Gamma}{A} & \cnlgseq{\Delta}{B}}$ &
$\infer[Cut]{\cnlgseq{\Gamma}{\Delta}}{\cnlgseq{\Delta}{A} & \cnlgseq{\Gamma}{A^{\bot}}}$ \\ \\
$\infer[\land^l]{\cnlgseq{\Gamma}{A\land B}}{\cnlgseq{\Gamma}{A}}$ &
$\infer[\land^r]{\cnlgseq{\Gamma}{A\land B}}{\cnlgseq{\Gamma}{B}}$ &
$\infer[\lor]{\cnlgseq{\Gamma}{A\lor B}}{\cnlgseq{\Gamma}{A} & \cnlgseq{\Gamma}{B}}$
\end{tabular}
\end{tabular}
\end{center}
\caption{A left-sided retelling of De Groote and Lamarche's original right-sided sequent calculus for \calccnl, augmented by rules for the additives.}
\label{cnl_degrlamstyle}
\end{figure}

\section{Focusing proofs}
Sequent calculi have seen widespread application in backward-chaining proof search: an attempt at proving a  \textit{goal} sequent proceeds by matching it against the conclusion of an inference rule, and replacing it by the latter's premises. The process terminates successfully once all goals have been replaced by axioms, and fails when there remain goals to prove while the applicable inferences have been exhausted. Cut elimination guarantees a reasonable bound on the search space: the only formulas found in the premises of the remaining inference rules already occur as subformulas of the conclusion.\footnote{The widespread appearance of $\cdot^{\bot}$ in the rules of Figure \ref{neutraldispl} necessitates a slightly nonconventional definition of the notion of subformula, as explicated in $\S$3.3.}

While satisfying the subformula property, Cut-free proof search still suffers from inessential non-determinism: neighboring logical inference steps involving different main and active formulas freely permute, making the choice of their relative ordering meaningless for settling the question of provability. The problem seeming inherent to the sequentialization of rule applications, Girard proposed a parallel representation of proofs, called \textit{proof nets}. In contrast, Andreoli stuck with sequent calculus, seeking instead a method for obtaining canonical representatives of derivations differing only by trivial rule permutations. Thus was born \textit{focused proof search}: greedily apply invertible inferences (preserving provability of the conclusion in the premises), while the active formulas appearing in the premises of non-invertible inferences always are to be principal. In other words, once chosen as main, a formula is `focused upon' in the sense of fixing the choice for subsequent rule applications to those targeting its subformulas.

The current section treats a succession of formalisms, each further realizing Andreoli's focusing strategy for \textbf{LG} and \textbf{CNL}. A brief review in $\S$3.1 of the causes for inessential nondeterminism in proof search reveals a partitioning of formulas into those of \textit{positive} or \textit{negative polarity}, depending on whether their inferences are always invertible. This leads in $\S$3.2 to an extension of the logical vocabulary by connectives for explicitly recording polarity shifts, with sequent derivations being adapted accordingly. The latter's normal (i.e., Cut-free) forms turn out to already satisfy \textit{weak} focalization, tackling permutations between non-invertible inferences. Full, or \textit{strong} focalization is obtained in $\S$3.3 through a sequent calculus of synthetic inferences (\cite{andreoli01}), collapsing multiple inference steps that freely permute.\footnote{The notions of \textit{strong} and \textit{weak} focalization as used here were, to the author's knowledge, first used in \cite{laurent04}.} The latter are furthermore compiled from the formulas appearing in the goal sequent(s), thus explicating the subformula property. Soundness and completeness w.r.t. unfocused provability, as discussed in $\S$2, are dealt with to the extent that all will be left to check is the normalization of polarized derivations into those considered strongly focalized. 

\subsection{Polarities and rule permutations}
In $\S$2.1, algebraic considerations led us to group the multiplicatives into the families $\{\lgtimes,/,\bs\}$ and $\{\lgplus,\os,\obs\}$, finding support in $\cdot^{\bot}$. 
Inspection of the logical rules in Figure \ref{neutraldispl}, 
however, 
reveals another natural classification:
\begin{center}
\begin{tabular}{rcll}
$P,Q$ & $::=$ & $p$ $|$ $(A\lgtimes B)$ $|$ $(A\os B)$ $|$ $(B\obs A)$ $|$ $(A\lor B)$ & Positive formulas \\
$N,M$ & $::=$ & $\bar{p}$ $|$ $(A\lgplus B)$ $|$ $(B\bs A)$ $|$ $(A/B)$ $|$ $(A\land B)$ & Negative formulas
\end{tabular}
\end{center}
Again, the dual of a positive under $\cdot^{\bot}$ is negative and vice versa, motivating the choice of terminology. Call a logical inference with a positive (negative) main formula positive (negative). We observe that positive inferences are always \textit{invertible}, meaning premises and conclusion may be interchanged. For example, invertibility of $(\obs)$ and (one half of) $(\land)$ is witnessed by the following Cuts:
\begin{center}
\begin{tabular}{ccc}
$\infer[T]{\cnlgseq{\Gamma}{\sobs{B^{\bot}}{A}}}{
\infer[/]{\cnlgseq{\sobs{B^{\bot}}{A}}{A^{\bot}/B^{\bot}}}{
\infer[I]{\cnlgseq{B^{\bot}}{B}}{} &
\infer[I]{\cnlgseq{A}{A^{\bot}}}{}} &
\cnlgseq{\Gamma}{B\obs A}}$ & \ \
$\infer[T]{\cnlgseq{\Gamma}{A}}{
\infer[\land^r]{\cnlgseq{A}{B^{\bot}\land A^{\bot}}}{
\infer[I]{\cnlgseq{A}{A^{\bot}}}{}} &
\cnlgseq{\Gamma}{A\lor B}}$
\end{tabular}
\end{center}
The positive/negative distinction provides a neat classification of rule permutations, always involving inference steps with disjoint active and main formulas:
\begin{enumerate}
\item Positive/negative. E.g., noting $\cnlgseq{B^{\bot}}{\stimes{C}{D}}$ iff $\cnlgseq{\stimes{C}{D}}{B^{\bot}}$ by $(Dp)$:
\begin{center}
\begin{tabular}{ccc}
$\infer=[Dp]{\cnlgseq{\stimes{A/B}{C\lgtimes D}}{\Gamma}}{
\infer[/]{\cnlgseq{\sobs{C\lgtimes D}{\Gamma}}{A/B}}{
\infer=[Dp]{\cnlgseq{C\lgtimes D}{B^{\bot}}}{
\infer[\lgtimes]{\cnlgseq{B^{\bot}}{C\lgtimes D}}{
\cnlgseq{B^{\bot}}{\stimes{C}{D}}}} &
\cnlgseq{\Gamma}{A}}}$ & &

$\infer=[Dp]{\cnlgseq{\stimes{A/B}{C\lgtimes D}}{\Gamma}}{
\infer[\lgtimes]{\cnlgseq{\sos{\Gamma}{A/B}}{C\lgtimes D}}{
\infer=[Dp]{\cnlgseq{\sos{\Gamma}{A/B}}{\stimes{C}{D}}}{
\infer[/]{\cnlgseq{\sobs{(\stimes{C}{D})}{\Gamma}}{A/B}}{
\cnlgseq{B^{\bot}}{\stimes{C}{D}} &
\cnlgseq{\Gamma}{A}}}}}$

\end{tabular}
\end{center}
\item Negative/negative. E.g., noting $\cnlgseq{C^{\bot}}{A}$ iff $\cnlgseq{A}{C^{\bot}}$ by $(Dp)$:
\begin{center}
\begin{tabular}{ccc}
$\infer=[Dp]{\cnlgseq{\stimes{(\stimes{A/B}{\Delta})}{C\bs D}}{\Gamma}}{
\infer[\bs]{\cnlgseq{\sos{\Gamma}{(\stimes{A/B}{\Delta})}}{C\bs D}}{
\infer=[Dp]{\cnlgseq{\stimes{A/B}{\Delta}}{C^{\bot}}}{
\infer[/]{\cnlgseq{\sobs{\Delta}{C^{\bot}}}{A/B}}{
\cnlgseq{\Delta}{B^{\bot}} &
\cnlgseq{C^{\bot}}{A}}} &
\cnlgseq{\Gamma}{D}}}$ & &

$\infer=[Dp]{\cnlgseq{\stimes{(\stimes{A/B}{\Delta})}{C\bs D}}{\Gamma}}{
\infer[/]{\cnlgseq{\sobs{\Delta}{(\sobs{C\bs D}{\Gamma})}}{A/B}}{
\infer=[Dp]{\cnlgseq{\sobs{C\bs D}{\Gamma}}{A}}{
\infer[\bs]{\cnlgseq{\sos{\Gamma}{A}}{C\bs D}}{
\cnlgseq{A}{C^{\bot}} &
\cnlgseq{\Gamma}{D}}} &
\cnlgseq{\Delta}{B^{\bot}}}}$
\end{tabular}
\end{center}

\item Positive/positive. E.g., noting $\cnlgseq{\sos{C}{D}}{\stimes{A}{B}}$ iff $\cnlgseq{\stimes{A}{B}}{\sos{C}{D}}$ by $(Dp)$:
\begin{center}
\begin{tabular}{ccc}
$\infer[\os]{\cnlgseq{A\lgtimes B}{C\os D}}{
\infer=[Dp]{\cnlgseq{A\lgtimes B}{\sos{C}{D}}}{
\infer[\lgtimes]{\cnlgseq{\sos{C}{D}}{A\lgtimes B}}{
\cnlgseq{\sos{C}{D}}{\stimes{A}{B}}}}}$ & &

$\infer=[Dp]{\cnlgseq{A\lgtimes B}{C\os D}}{
\infer[\lgtimes]{\cnlgseq{C\os D}{A\lgtimes B}}{
\infer=[Dp]{\cnlgseq{C\os D}{\stimes{A}{B}}}{
\infer[\os]{\cnlgseq{\stimes{A}{B}}{C\os D}}{
\cnlgseq{\stimes{A}{B}}{\sos{C}{D}}}}}}$
\end{tabular}
\end{center}

\end{enumerate}
Besides the reorderings caused by the logical rules, the current sequent calculus is home to an additional form of redundancy, courtesy of the display postulates. Recall the latter's purpose is to isolate the main formula of a logical inference from within a presentation. While always possible in a canonical fashion, nothing prevents us from taking detours, e.g., displaying formulas without applying the corresponding logical inference, or even introducing cycles (revisiting the same sequent multiple times throughout a derivation). Thus, we wish to constrain their applicability, ideally doing away with them altogether. 

In the sequel, we consider variations of \calccnl \ and \calclg \ that internalizes polarity shifts within the syntax of formulas. We define the corresponding sequent derivations in $\S$3.2, these being \textit{two-sided} in the sense that all invertible inferences apply on the left-hand side, while all non-invertible, formerly negative inferences apply on the right-hand. Their Cut-free forms address the negative/negative permutations by requiring the active formulas of negative inferences to be principal: once a main formula is chosen, at least stick with it. Positive/positive permutations benefit from a sequent presentation involving synthetic inferences ($\S$3.3), collapsing into a single rule the stepwise decomposition of a positive formula into its negative subparts, while finally positive/negative permutations are eliminated by enforcing the greedy application of positive inferences: when deciding provability, a positive step is always the safer choice. 

\subsection{Polarized sequent derivations}
We define \textit{polarized} \textbf{CNL} and \textbf{LG}, referred to by \cnlpol \ and \lgpol \ respectively, their logical vocabulary containing additional connectives for recording polarity shifts.

\begin{defs} Formulas are made inherently positive or negative, with \textit{shifts} $\pton,\ntop$ (\cite{girard00}, $\S$5.3) establishing communication:
\begin{center}
\begin{tabular}{rcl}
$P,Q$ & $::=$ & $p$ $|$ $(P\lgtimes Q)$ $|$ $(P\os N)$ $|$ $(N\obs P)$ $|$ $(P\lor Q)$ $|$ $(\ntop N)$ \\
$N,M$ & $::=$ & $\bar{p}$ $|$ $(M\lgplus N)$ $|$ $(Q\bs M)$ $|$ $(M/Q)$ $|$ $(M\land N)$ $|$ $(\pton P)$
\end{tabular}
\end{center}
Linear negation $\cdot^{\bot}$ is revised accordingly, satisfying $P^{\bot\bot}=P$ and $N^{\bot\bot}=N$:
\begin{center}
\begin{tabular}{rclcrcl}
$p^{\bot}$ & $\defeq$ & $\bar{p}$ & \ \ & $\bar{p}^{\bot}$ & $\defeq$ & $p$ \\
$(P\lgtimes Q)^{\bot}$ & $\defeq$ & $Q^{\bot}\lgplus P^{\bot}$ & & $(M\lgplus N)^{\bot}$ & $\defeq$ & $N^{\bot}\lgtimes M^{\bot}$ \\
$(N\obs P)^{\bot}$ & $\defeq$ & $P^{\bot}/N^{\bot}$ & & $(M/Q)^{\bot}$ & $\defeq$ & $Q^{\bot}\obs M^{\bot}$ \\
$(P\os N)^{\bot}$ & $\defeq$ & $N^{\bot}\bs P^{\bot}$ & & $(Q\bs M)^{\bot}$ & $\defeq$ & $M^{\bot}\os Q^{\bot}$ \\
$(P\lor Q)^{\bot}$ & $\defeq$ & $Q^{\bot}\land P^{\bot}$ & & $(M\land N)^{\bot}$ & $\defeq$ & $N^{\bot}\lor M^{\bot}$ \\
$(\ntop N)^{\bot}$ & $\defeq$ & $\pton N^{\bot}$ & & $(\pton P)^{\bot}$ & $\defeq$ & $\ntop P^{\bot}$
\end{tabular}
\end{center}
\end{defs}

\begin{rem} In practice, we assume $\pton,\ntop$ to bind more strongly than $\lgtimes,\lgplus$, and drop brackets accordingly. Thus, $P\lgtimes\pton Q^{\bot}$ abbreviates $(P\lgtimes(\pton Q^{\bot}))$.
\end{rem}

\begin{defs}\label{polar_struc}
With our logical vocabulary extended by shifts, we revise structures and presentations so as to contain only positive formulas:
\begin{center}
\begin{tabular}{rcl}
$\svara,\svarb$ & $::=$ & $P$ $|$ $(\stimes{\svara}{\svarb})$ $|$ $(\sos{\svara}{\svarb})$ $|$ $(\sobs{\svarb}{\svara})$ \\
$\omega$ & $::=$ & $\pres{\svara}{\svarb}$
\end{tabular}
\end{center}
In particular, antecedent negative formulas appear only as $\ntop N$.
\end{defs}

\begin{defs}\label{polar_struc_form}
As before, structures $\svara$ interpret by dual formulas $\svara^{\inp}$ and $\svara^{\outp}$:
\begin{center}
\begin{tabular}{rclcrcl}
$P^{\inp}$ & $\defeq$ & $P$ & \ \ & $P^{\outp}$ & $\defeq$ & $P^{\bot}$ \\
$(\stimes{\svara}{\svarb})^{\inp}$ & $\defeq$ & $\svara^{\inp}\lgtimes\svarb^{\inp}$ & & $(\stimes{\svara}{\svarb})^{\outp}$ & $\defeq$ & $\svarb^{\outp}\lgplus\svara^{\outp}$ \\
$(\sobs{\svarb}{\svara})^{\inp}$ & $\defeq$ & $\svarb^{\outp}\obs\svara^{\inp}$ & & $(\sobs{\svarb}{\svara})^{\outp}$ & $\defeq$ & $\svara^{\outp}/\svarb^{\inp}$ \\
$(\sos{\svara}{\svarb})^{\inp}$ & $\defeq$ & $\svara^{\inp}\os\svarb^{\outp}$ & & $(\sos{\svara}{\svarb})^{\outp}$ & $\defeq$ & $\svarb^{\inp}\bs\svara^{\outp}$
\end{tabular}
\end{center}
\end{defs}

\begin{defs}
The sequent calculus for \textbf{LG}$_{\emptyset}^{\textit{pol}}$ is provided in F.\ref{cnlg_weakfoc}, involving derivability judgements $\presa\vdash$ and $\cnlgseqp{\svara}{P}$ defined by mutual induction. We refer by the \textit{stoup} to the righthand side of the turnstile, adapting terminology of (\cite{girard91}). The extension by structural rules remains unchanged from Figures \ref{neutraldispl_grishin} and \ref{neutraldispl_cnl}.
\end{defs}

\begin{figure}
\begin{center}
\begin{tabular}{ccc}
 & $\infer[T]{\cnlgseq{\svara}{\svarb}}{\cnlgseqp{\svarb}{P} & \cnlgseq{\svara}{P}}$ \\ \\
$\infer[I]{\cnlgseqp{P}{P}}{}$ & $\infer[\ntop L]{\cnlgseq{\svara}{\ntop N}}{\cnlgseqp{\svara}{N^{\bot}}}$ & $\infer[\ntop R]{\cnlgseqp{\svara}{\ntop N}}{\cnlgseq{\svara}{N^{\bot}}}$ \\ \\
$\infer[\lgtimes L]{\cnlgseq{\svara}{P\lgtimes Q}}{\cnlgseq{\svara}{\stimes{P}{Q}}}$ &
$\infer[\obs L]{\cnlgseq{\svara}{N\obs P}}{\cnlgseq{\svara}{\sobs{N^{\bot}}{P}}}$ &
$\infer[\os L]{\cnlgseq{\svara}{P\os N}}{\cnlgseq{\svara}{\sos{P}{N^{\bot}}}}$ \\ \\
$\infer[\lgtimes R]{\cnlgseqp{\stimes{\svara}{\svarb}}{P\lgtimes Q}}{\cnlgseqp{\svara}{P} & \cnlgseqp{\svarb}{Q}}$ &
$\infer[\obs R]{\cnlgseqp{\sobs{\svarb}{\svara}}{N\obs P}}{\cnlgseqp{\svarb}{N^{\bot}} & \cnlgseqp{\svara}{P}}$ &
$\infer[\os R]{\cnlgseqp{\sos{\svara}{\svarb}}{P\os N}}{\cnlgseqp{\svarb}{N^{\bot}} & \cnlgseqp{\svara}{P}}$ \\ \\
$\infer[\lor L]{\cnlgseq{\svara}{P\lor Q}}{\cnlgseq{\svara}{P} & \cnlgseq{\svara}{Q}}$ &
$\infer[\lor R^l]{\cnlgseqp{\Gamma}{P\lor Q}}{\cnlgseqp{\Gamma}{P}}$ &
$\infer[\lor R^r]{\cnlgseqp{\Gamma}{P\lor Q}}{\cnlgseqp{\Gamma}{Q}}$ \\ \\
$\infer[dp]{\cnlgseq{\svara}{\svarb}}{}$ & 
$\infer=[dp]{\cnlgseq{\svara}{\sobs{\svarb}{\svarc}}}{\cnlgseq{\stimes{\svara}{\svarb}}{\svarc}}$ &
$\infer=[dp]{\cnlgseq{\sos{\svara}{\svarb}}{\svarc}}{\cnlgseq{\svara}{\stimes{\svarb}{\svarc}}}$
\end{tabular}
\end{center}
\caption{Polarized sequent calculus: base logic.}
\label{cnlg_weakfoc}
\end{figure}

\noindent Note that Figure \ref{cnlg_weakfoc} makes no explicit mention of rules for deriving negative formulas. Instead, these are hidden inside the right introductions, as is clear when the latter precede an application of $(\ntop L)$:
\begin{center}
$\infer[\ntop L]{\cnlgseq{\svarb}{\ntop N}}{
\cnlgseqp{\svarb}{N^{\bot}}}$
\end{center}
Thus, a right introduction of a positive formula may be understood as a left introduction of its negative dual. This intuition is further pursued through a completeness proof w.r.t. the sequent derivations of $\S$2, demonstrated using the following decoration of unpolarized formulae with shifts.

\begin{defs}\label{shift_decoration} For $A$ a(n unpolarized) formula, let $\epsilon(A)=+$ if $A$ is of positive polarity, i.e., of the form $p$, $B\lgtimes C$, $B\os C$ or $C\obs B$, and $\epsilon(A)=-$ otherwise. We translate $A$ into a formula $\semt{A}$ of \cnlpol/\lgpol \ avoiding `vacuous' polarity shifts by excluding
subformulas of the form $\ntop\pton P$ or $\pton\ntop N$. For the base cases, we define $\semt{p}=p$ and $\semt{\bar{p}}=\bar{p}$. For complex $A$ with $\epsilon(A)=+$, we set
\begin{center}
\begin{tabular}{|@{}c@{}c@{}|@{}l@{}|@{}l@{}|@{}l@{}|@{}l@{}|}\hline
 \ $\epsilon(A)$ \ & $\epsilon(B)$ \ & \ $\semt{A\lgtimes B}$ & \ $\semt{A\os B}$ & \ $\semt{B\obs A}$ & \ $A\lor B$ \\ \hline
$+$ & $+$ & \ $\semt{A}\lgtimes\semt{B}$ & \ $\semt{A}\os\pton\semt{B}$ & \ $\pton\semt{B}\obs\semt{A}$ & \ $\semt{A}\lor\semt{B}$ \\ 
$+$ & $-$ & \ $\semt{A}\lgtimes\ntop\semt{B}$ & \ $\semt{A}\os\semt{B}$ & \ $\semt{B}\obs\semt{A}$ & \ $\semt{A}\lor\ntop\semt{B}$ \\
$-$ & $+$ & \ $\ntop\semt{A}\lgtimes\semt{B}$ & \ $\ntop\semt{A}\os\pton\semt{B}$ & \ $\pton\semt{B}\obs\ntop\semt{A}$ & \ $\ntop\semt{A}\lor\semt{B}$ \\
$-$ & $-$ & \ $\ntop\semt{A}\lgtimes\ntop\semt{B}$ & \ $\ntop\semt{A}\os\semt{B}$ & \ $\semt{B}\obs\ntop\semt{A}$ & \ $\ntop\semt{A}\lor\ntop\semt{B}$ \\ \hline
\end{tabular}
\end{center}
Finally, for complex $A$ with $\epsilon(A)=-$,
\begin{center}
\begin{tabular}{|@{}c@{}c@{}|@{}l@{}|@{}l@{}|@{}l@{}|@{}l@{}|}\hline
 \ $\epsilon(A)$ \ & $\epsilon(B)$ \ & \ $\semt{A\lgplus B}$ & \ $\semt{B\bs A}$ & \ $\semt{A/B}$ & \ $A\land B$ \\ \hline
$+$ & $+$ & \ $\pton\semt{A}\lgplus\pton\semt{B}$ & \ $\semt{B}\bs\pton\semt{A}$ & \ $\pton\semt{A}/\semt{B}$ & \ $\pton\semt{A}\land\pton\semt{B}$ \\ 
$+$ & $-$ & \ $\pton\semt{A}\lgplus\semt{B}$ & \ $\ntop\semt{N}\bs\pton\semt{A}$ & \ $\pton\semt{A}/\ntop\semt{A}$ & \ $\pton\semt{A}\land\semt{B}$ \\
$-$ & $+$ & \ $\semt{A}\lgplus\pton\semt{B}$ & \ $\semt{B}\bs\semt{A}$ & \ $\semt{A}/\semt{B}$ & \ $\semt{A}\land\pton\semt{B}$ \\
$-$ & $-$ & \ $\semt{A}\lgplus\semt{B}$ & \ $\ntop\semt{B}\bs\semt{A}$ & \ $\semt{A}/\pton\semt{B}$ & \ $\semt{A}\land\semt{B}$ \\ \hline
\end{tabular}
\end{center}
\end{defs}
\noindent An easy induction establishes
\begin{lem}\label{semt_lneg}
The map $\semt{\cdot}$ commutes with linear negation. I.e., for any $A$, $\semt{A}^{\bot}=\semt{A^{\bot}}$.
\end{lem}
\begin{defs}
We extend $\semt{\cdot}$ to the level of structures $\Gamma$ as follows:
\begin{center}
\begin{tabular}{ccc}
\begin{tabular}{c}$A\mapsto\left\{\begin{array}{rl}\semt{A} & \textrm{ if }\epsilon(A)=+ \\ \ntop\semt{A} & \textrm{ if }\epsilon(A)=-\end{array}\right.$\end{tabular} & &
$\begin{array}{rcl}\semt{\stimes{\Gamma}{\Delta}} & \defeq & \stimes{\semt{\Gamma}}{\semt{\Delta}} \\
\semt{\sos{\Gamma}{\Delta}} & \defeq & \sos{\semt{\Gamma}}{\semt{\Delta}} \\
\semt{\sobs{\Delta}{\Gamma}} & \defeq & \sobs{\semt{\Delta}}{\semt{\Gamma}}\end{array}$
\end{tabular}
\end{center}
\end{defs}

\noindent Below, we show $\cnlgseq{\svara}{\svarb}$ in \calclg \ (\calccnl) implies $\cnlgseq{\semt{\svara}}{\semt{\svarb}}$ in \lgpol \ (\cnlpol), making repeated use of Cut. As a consequence, we claim that a proof of the latter's admissibility suffices for showing a weak form of focalization, restricting negative/negative permutations, though leaving the others unaddressed. Indeed, all non-invertible inferences take place entirely within the stoup, housing at most a single formula. Since, in the absence of Cut, formulas are allowed to cross the turnstile only upon the encounter of polarity shifts, and `vacuous' such shifts through subformulas $\ntop\pton P$ and $\pton\ntop N$ have been avoided, maximal chains of non-invertible inferences are enforced, traversing the formula tree of the particular $N$ designated main through $(\ntop L)$.

\begin{thm}\label{unpol_weak} If $\cnlgseq{\Gamma}{\Delta}$ in \calclg/\calccnl, then $\cnlgseq{\semt{\Gamma}}{\semt{\Delta}}$ in \lgpol/\cnlpol.
\end{thm}
\begin{proof}
By induction on $\cnlgseq{\Gamma}{\Delta}$, making free use of L.\ref{semt_lneg}. The structural rules, including the display postulates, are immediate. In each remaining case, we must consider all possible values for $\epsilon(A_1),\dots,\epsilon(A_n)$ of the active and main formulas $A_1,\dots,A_n$ involved. For axioms and Cut, we have
\begin{center}
\begin{tabular}{ccc}
$\begin{array}{c}\infer[\ntop L]{\cnlgseq{\semt{A}}{\ntop\semt{A^{\bot}}}}{
\infer[I]{\cnlgseqp{\semt{A}}{\semt{A}}}{}}\end{array}$ &
$\begin{array}{c}\infer[T]{\cnlgseq{\semt{\svara}}{\semt{\svarb}}}{
\infer[\ntop I]{\cnlgseqp{\svarb}{\ntop\semt{A}^{\bot}}}{
\cnlgseq{\svarb}{\semt{A}}} &
\cnlgseq{\semt{\svara}}{\ntop\semt{A^{\bot}}}}\end{array}$ &
(if $\epsilon(A)=+$) \\ \\
$\begin{array}{c}\infer=[Dp]{\cnlgseq{\ntop\semt{A}}{\semt{A^{\bot}}}}{
\infer[\ntop L]{\cnlgseq{\semt{A^{\bot}}}{\ntop\semt{A}}}{
\infer[I]{\cnlgseqp{\semt{A^{\bot}}}{\semt{A^{\bot}}}}{}}}\end{array}$ &
$\begin{array}{c}\infer=[Dp]{\cnlgseq{\semt{\svara}}{\semt{\svarb}}}{
\infer[T]{\cnlgseq{\semt{\svarb}}{\semt{\svara}}}{
\infer[\ntop I]{\cnlgseqp{\svara}{\ntop\semt{A}}}{
\cnlgseq{\svara}{\semt{A^{\bot}}}} &
\cnlgseq{\semt{\svarb}}{\ntop\semt{A}}}}\end{array}$ &
(if $\epsilon(A)=-$)
\end{tabular}
\end{center}
Each of the cases $(\lgplus)$, $(/)$ and $(\bs)$ are handled similarly, so we suffice by checking $(\lgplus)$. We consider the situation $\epsilon(A)=\epsilon(B)=-$:  
\begin{center}
\begin{tabular}{l}
$\infer=[Dp]{\cnlgseq{\stimes{\semt{\svarb}}{\semt{\svara}}}{\ntop(\semt{A}\lgplus\semt{B})}}{
\infer[T]{\cnlgseq{\sos{\ntop(\semt{A}\lgplus\semt{B})}{\semt{\svarb}}}{\semt{\svara}}}{
\infer[\ntop R]{\cnlgseqp{\sos{\ntop(\semt{A}\lgplus\semt{B})}{\semt{\svarb}}}{\ntop\semt{A}}}{
\infer=[Dp]{\cnlgseq{\sos{\ntop(\semt{A}\lgplus\semt{B})}{\semt{\svarb}}}{\semt{A}^{\bot}}}{
\infer[T]{\cnlgseq{\sobs{\semt{A}^{\bot}}{\ntop(\semt{A}\lgplus\semt{B})}}{\semt{\svarb}}}{
\infer[\ntop R]{\cnlgseqp{\sobs{\semt{A}^{\bot}}{\ntop(\semt{A}\lgplus\semt{B})}}{\ntop\semt{B}}}{
\infer=[Dp]{\cnlgseq{\sobs{\semt{A}^{\bot}}{\ntop(\semt{A}\lgplus\semt{B})}}{\semt{B}^{\bot}}}{
\infer[\ntop L]{\cnlgseq{\stimes{\semt{B}^{\bot}}{\semt{A}^{\bot}}}{\ntop(\semt{A}\lgplus\semt{B})}}{
\infer[\lgtimes R]{\cnlgseqp{\stimes{\semt{B}^{\bot}}{\semt{A}^{\bot}}}{\semt{B}^{\bot}\lgtimes\semt{A}^{\bot}}}{
\infer[I]{\cnlgseqp{\semt{B}^{\bot}}{\semt{B}^{\bot}}}{} &
\infer[I]{\cnlgseqp{\semt{A}^{\bot}}{\semt{A}^{\bot}}}{}}}}} &
\cnlgseq{\semt{\svarb}}{\ntop\semt{B}}}}} &
\cnlgseq{\semt{\svara}}{\ntop\semt{A}}}}$
\end{tabular}
\end{center}
The remaining cases $(\lgtimes)$, $(\os)$ and $(\obs)$ trivially translate to left introductions.
\end{proof}

\subsection{Strong focalization}
Our efforts so far have left all but negative/negative permutations unaddressed. Our interest, however, is in full, or strong focalization, providing further normalization of Cut-free polarized derivations. In particular, we structure topdown proof-search into alternating invertible and non-invertible phases. The latter proceeds as in the previous section, taking place entirely within the stoup, and ends in each branch with an application of $(\ntop R)$. All applicable invertible inferences are subsequently to be exhausted before $(\ntop L)$ is made available, ushering in a new non-invertible phase. By furthermore collapsing both phases into single inference steps, the problem posed by permutations between the inferences within a single invertible phase is remedied. 

Despite offering increased control over rule ordering, focused derivations remain sequent derivations at heart, satisfying in particular a local subformula property. From the above discussion, however, it should be clear that subformulas are identified only up to polarity shifts, marking the boundaries between the invertible and non-invertible phases.


\begin{defs}\label{useful_subform}
For any $P$ or $N$, the maps $\sigma$ and $\tau$, defined by mutual induction, pick out the subformulas relevant for focused proof search:
\begin{center}
\begin{tabular}{rclcrcl}
$\suba{\pton P}$ & $\defeq$ & $\{ P^{\bot}\}\bigcup\subb{P}$ & &
$\subb{\ntop N}$ & $\defeq$ & $\{ N\}\bigcup\suba{N}$ \\
$\suba{\bar{p}}$ & $\defeq$ & $\{\bar{p}\}$ & &
$\subb{p}$ & $\defeq$ & $\{ p\}$ \\
$\suba{M\land N}$ & $\defeq$ & $\suba{M}\bigcup\suba{N}$ & &
$\subb{P\lor Q}$ & $\defeq$ & $\subb{P}\bigcup\subb{Q}$ \\
$\suba{M\lgplus N}$ & $\defeq$ & $\suba{M}\bigcup\suba{N}$ & &
$\subb{P\lgtimes Q}$ & $\defeq$ & $\subb{P}\bigcup\subb{Q}$ \\
$\suba{M/Q}$ & $\defeq$ & $\suba{M}\bigcup\subb{Q}$ & &
$\subb{N\obs P}$ & $\defeq$ & $\subb{P}\bigcup\suba{N}$ \\
$\suba{Q\bs M}$ & $\defeq$ & $\suba{M}\bigcup\subb{Q}$ & &
$\subb{P\os N}$ & $\defeq$ & $\subb{P}\bigcup\suba{N}$
\end{tabular}
\end{center}
\end{defs}
\noindent Note carefully the definition of $\suba{\pton P}$ as $\{ P^{\bot}\}\bigcup\subb{P}$ instead of $\{ P\}\bigcup\subb{P}$. Explained in terms of the sequent derivations of $\S$3.2, this is due to the formulation of the rules for $\ntop$: since in the premise of $(\ntop L)$ the stoup contains the negation of the main formula $\ntop N$, the latter's subformulas of the form $\pton P$ also appear as $\ntop P^{\bot}$ when main in an instance of $(\ntop R)$. Thus, it is $P^{\bot}$, and not $P$, that we wish to remember at this particular polarity shift. As a further indication of the harmony of this definition, a straightforward induction proves
\begin{lem}
For any $N,P$, $\suba{N}=\subb{N^{\bot}}$ and (dually) $\subb{P}=\suba{P^{\bot}}$.
\end{lem} 
\noindent In what is to follow, fix a set $X$ of negative formulas. 
Concepts are defined relative to $X$, its intended instantiation being as follows. The aim of this section being the normalization of polarized derivations, we fix $X$ by $\{\Gamma^-,\Delta^-\}$ for an initial goal sequent $\cnlgseq{\Gamma}{\Delta}$ according to D.\ref{polar_struc}. This determines a \textit{set} of goal sequents, as further elaborated upon below, s.t. the latter's strongly focused provability guarantees the provability of the initial $\cnlgseq{\Gamma}{\Delta}$ (cf. C.\ref{final_main_theorem}). While not a necessary ingredient for the definitions to follow, being easily ignorable, we find the explicit parameterization over such sets better emphasizes the goal-driven nature of the current take on focalization, while furthermore serving as an explicit check on the satisfaction of the subformula property. In particular, inference rules are defined only for the members of the closure 
$X^{\tau}\defeq\{\subb{\ntop N} \ | N\in X\}$ of $X$ under $\subb{\cdot}$.\footnote{Bottom-up variations on focused proof search have also been considered by Chaudhuri and Pfenning (\cite{chaudhuripfenning}) in the context of linear logic.} The following is a simple case analysis.
\begin{lem}
If $N\in X^{\tau}$, then $\subb{\ntop N}\subseteq X^{\tau}$.
\end{lem}
\begin{defs}
Structures are revised so as to prohibit positive formulas other than of the form $\ntop N$ or $p$. In the former case, we can leave the shift $\ntop$ implicit, arriving at the following definition, where $p,N\in X^{\tau}$ in the base cases:
\begin{center}
\begin{tabular}{rcl}
$\svard,\svare,\svarf$ & $::=$ & $p$ $|$ $N$ $|$ $(\stimes{\svard}{\svare})$ $|$ $(\sos{\svard}{\svare})$ $|$ $(\sobs{\svare}{\svard})$ \\
$\omega$ & $::=$ & $\svard,\svare$
\end{tabular}
\end{center}
The interpretation of structures $\svard$ by dual formulas $\svard^{\inp}$ and $\svard^{\outp}$ is a straightforward adaption of D.\ref{polar_struc_form}, where in particular $N^{\inp}=N$ and $N^{\outp}=N^{\bot}$. Note we do not require $\svard^{\inp},\svard^{\outp}\in X^{\tau}$.
\end{defs}
\noindent Since the logical vocabulary remains unchanged from $\S$3.2, we have chosen not to overload notation any further and use different metavariables for denoting structures to prevent confusion.

\begin{defs}
To absorb the display postulates, we resort to the use of \textit{contexts} $\presa[\svare]$, representing presentations with a distinguished occurrence of $\svarb$.
\begin{center}
\begin{tabular}{rclrcl}
$\svard[],\svare[]$ & $::=$ & $[]$ $|$ $(\stimes{\svard[]}{\svare})$ $|$ $(\stimes{\svard}{\svare[]})$ & 
$\presa[]$ & $::=$ & $\pres{\svard[]}{\svare}$ $|$ $\pres{\svard}{\svare[]}$ \\
 & $|$ & $(\sos{\svard[]}{\svare})$ $|$ $(\sos{\svard}{\svare[]})$ \\
 & $|$ & $(\sobs{\svare}{\svard[]})$ $|$ $(\sobs{\svare[]}{\svard})$
\end{tabular}
\end{center}
Let $\svard[\svare]$ ($\presa[\svare]$) denote the result of substituting $\svare$ for the unique occurrence of $[]$ in $\svard[]$ ($\presa[]$). For $\svard[]$ and $\svare[]$ contexts, we denote by $\svard[\svare[]]$ their composition, with insertion of $\svarf$ understood as the insertion of $\svare[\svarf]$ in $\svard[]$.
\end{defs}
\begin{defs}
The map $\div$ takes pairs of contexts and structures into structures, being defined by induction over its first argument:
\begin{center}
\begin{tabular}{r@{}c@{}l@{}c@{}r@{}c@{}l}
$\displ{[]}{\svarf}$ & \ \ $\defeq$ \ \ & $\svarf$ \\
$\displ{(\stimes{\svard[]}{\svare})}{\svarf}$ & \ $\defeq$ \ & $\displ{\svard[]}{(\sobs{\svare}{\svarf})}$ & \ \ \ \ &
$\displ{(\stimes{\svard}{\svare[]})}{\svarf}$ & \ \ $\defeq$ \ \ & $\displ{\svare[]}{(\sos{\svarf}{\svard})}$ \\
$\displ{(\sos{\svard[]}{\svare})}{\svarf}$ & \ $\defeq$ \ & $\displ{\svard[]}{(\stimes{\svare}{\svarf})}$ & &
$\displ{(\sobs{\svare}{\svard[]})}{\svarf}$ & \ $\defeq$ \ & $\displ{\svard[]}{(\stimes{\svarf}{\svare})}$ \\
$\displ{(\sos{\svard}{\svare[]})}{\svarf}$ & \ $\defeq$ \ & $\displ{\svard[]}{(\sobs{\svarf}{\svard})}$ & &
$\displ{(\sobs{\svare[]}{\svard})}{\svarf}$ & \ $\defeq$ \ & $\displ{\svard[]}{(\sos{\svard}{\svarf})}$
\end{tabular}
\end{center}
The intuition we pursue is that $\cnlgseq{\svard[\svare]}{\svarf}$ iff $\cnlgseq{\displ{\svard[]}{\svarf}}{\svare}$ through the display postulates. In particular, defining $\presa^*[]$ by $\displ{\svard[]}{\svare}$ for any $\presa[]=\pres{\svard[]}{\svare}$ or $\presa[]=\pres{\svare}{\svard[]}$, $\presa[\svard]\vdash$ iff $\cnlgseq{\presa^*[]}{\svard}$.
\end{defs}

\begin{defs}\label{def_invp}
For each positive $P$, the set $\invp{P}$ decomposes $P$ into its structural counterparts:
\begin{center}
\begin{tabular}{rclrcl}
$\invp{P\lgtimes Q}$ & $\defeq$ & $\{\stimes{\svard}{\svare} \ | \ \svard\in\invp{P}, \ \svare\in\invp{Q}\}$ &
$\invp{P\lor Q}$ & $\defeq$ & $\invp{P}\mathbin{\bigcup}\invp{Q}$ \\
$\invp{P\os N}$ & $\defeq$ & $\{\sos{\svard}{\svare} \ | \ \svard\in\invp{P}, \ \svare\in\invp{N^{\bot}}\}$ &
$\invp{p}$ & $\defeq$ & $\{ p\}$ \\
$\invp{N\obs P}$ & $\defeq$ & $\{\sobs{\svare}{\svard} \ | \ \svard\in\invp{P}, \ \svare\in\invp{N^{\bot}}\}$ &
$\invp{\ntop N}$ & $\defeq$ & $\{ N\}$
\end{tabular}
\end{center}
One easily shows that if $N\in X^{\tau}$, then $\invp{N^{\bot}}$ is well-defined relative to $X^{\tau}$.
\end{defs}

\noindent Figure \ref{cnlg_strongfoc_approx} provides a first approximation of strong focalization. Compared to Figure \ref{cnlg_weakfoc}, invertible inferences are compiled away into the right introduction of $\ntop$, ensuring their greedy application. Roughly, the inference of $\ntop N$ requires a premise for each element of $\invp{N^{\bot}}$, calling to attention the fact that the only branching left introductions of Figure \ref{cnlg_weakfoc} are those introducing additives, and similarly $\invp{N^{\bot}}$ is a singleton iff no additives are encountered in $N^{\bot}$ up to the first immediate polarity switches. Save for the above revision of $(\ntop R)$ and the renaming of $(\ntop L)$ into decisions $(D)$, right introductions remain unaltered from F.\ref{cnlg_weakfoc}, violating our faithfulness to $X^{\tau}$. What is needed is a reformulation of $(D)$ so as to take the sets $\invp{P}$ into account, just like we did for the invertible phase.

\begin{figure}
\begin{center}
\begin{tabular}{ccc}
\multicolumn{3}{c}{
$\infer[D]{\presa[N]\vdash}{\cnlgseqp{\presa^*[]}{N^{\bot}}}$ \ \ \ \ $\infer[\ntop]{\cnlgseqp{\svard}{\ntop N}}{\{\cnlgseq{\svard}{\svare} \ | \ \svare\in\invp{N^{\bot}}\}}$} \\ \\
$\infer[I]{\cnlgseqp{p}{p}}{}$ &
$\infer[\lor^l]{\cnlgseqp{\Gamma}{P\lor Q}}{\cnlgseqp{\Gamma}{P}}$ &
$\infer[\lor^r]{\cnlgseqp{\Gamma}{P\lor Q}}{\cnlgseqp{\Gamma}{Q}}$ \\ \\
$\infer[\lgtimes]{\cnlgseqp{\stimes{\svard}{\svare}}{P\lgtimes Q}}{\cnlgseqp{\svard}{P} & \cnlgseqp{\svare}{Q}}$ &
$\infer[\obs]{\cnlgseqp{\sobs{\svare}{\svard}}{N\obs P}}{\cnlgseqp{\svare}{N^{\bot}} & \cnlgseqp{\svard}{P}}$ &
$\infer[\os]{\cnlgseqp{\sos{\svard}{\svare}}{P\os N}}{\cnlgseqp{\svare}{N^{\bot}} & \cnlgseqp{\svard}{P}}$ 
\end{tabular}
\end{center}
\caption{A first approximation of strongly focalized derivations.}
\label{cnlg_strongfoc_approx}
\end{figure}

\begin{defs} Figure \ref{cnlg_strongfoc} defines strong normalization for \calclg$_{\emptyset}$, involving judgements $\cnlgseq{\svard}{\svare}$ and $\cnlgseqp{\svard}{\svare}$. The latter addresses context splitting during the non-invertible phase (i.e., the distribution of a structure appearing in the conclusion over the premises), and replaces the previous judgement form $\cnlgseqp{\svard}{P}$. In particular, right introductions of $\ntop$ have been renamed Reactions $(R)$, while $(\stimes{}{})$, $(\sos{}{})$ and $(\sobs{}{})$ resemble $(\lgtimes R)$, $(\os R)$ and $(\obs R)$ respectively. The remaining $(\lor R^l)$ and $(\lor R^r)$ are compiled away into the revised Decisions $(D)$. Optional structural extensions are listed in Figure \ref{cnlg_strongfoc_struc}.
\end{defs}

\begin{figure}
\begin{center}
\begin{tabular}{ccc}
\multicolumn{3}{c}{$\begin{array}{cl}\begin{array}{c}\infer[D]{\presa[N]\vdash}{\cnlgseqp{\presa^*[]}{\svare}}\end{array} & \textrm{for any }N\in X^{\tau}\textrm{, }\svare\in\invp{N^{\bot}} \\ \\
\infer[I]{\cnlgseqp{p}{p}}{} &
\infer[R]{\cnlgseqp{\svard}{N}}{\{\cnlgseq{\svard}{\svare} \ | \ \svare\in\invp{N^{\bot}}\}}
\end{array}$} \\ \\
$\infer[\stimes{}{}]{\cnlgseqp{\stimes{\svard}{\svare}}{\stimes{\svard'}{\svare'}}}{\cnlgseqp{\svard}{\svard'} & \cnlgseqp{\svare}{\svare'}}$ &
$\infer[\sobs{}{}]{\cnlgseqp{\sobs{\svare}{\svard}}{\sobs{\svare'}{\svard'}}}{\cnlgseqp{\svare}{\svare'} & \cnlgseqp{\svard}{\svard'}}$ &
$\infer[\sos{}{}]{\cnlgseqp{\sos{\svard}{\svare}}{\sos{\svard'}{\svare'}}}{\cnlgseqp{\svare}{\svare'} & \cnlgseqp{\svard}{\svard'}}$ 
\end{tabular}
\end{center}
\caption{Strongly focalized derivations: base logic.}
\label{cnlg_strongfoc}
\end{figure}

\begin{figure}[h]
\begin{center}
\begin{tabular}{ccc}
\multicolumn{3}{c}{CNL} \\ \\
$\infer=[\cnlos]{\cnlgseqp{\svard[\sos{\svare_1}{\svare_2}]}{\svarf}}{\cnlgseqp{\svard[\stimes{\svare_1}{\svare_2}]}{\svarf}}$ & &
$\infer=[\cnlobs]{\cnlgseqp{\svard[\sobs{\svare_1}{\svare_2}]}{\svarf}}{\cnlgseqp{\svard[\stimes{\svare_1}{\svare_2}]}{\svarf}}$ \\ \\
\multicolumn{3}{c}{Linear Distributivity} \\ \\
$\infer[\grisha{I}{1a}]{\cnlgseqp{\svard[\sos{(\stimes{\svarb_1}{\svarb_2})}{\svarb_3}]}{\svarf}}{\cnlgseqp{\svard[\stimes{\svarb_1}{(\sos{\svarb_2}{\svarb_3})}]}{\svarf}}$ & &
$\infer[\grisha{I}{1b}]{\cnlgseqp{\svard[\sobs{\svarb_1}{(\stimes{\svarb_2}{\svarb_3})}]}{\svarf}}{
\cnlgseqp{\svard[\sobs{(\sos{\svarb_1}{\svarb_2})}{\svarb_3}]}{\svarf}}$ \\ \\
$\infer[\grisha{I}{2a}]{\cnlgseqp{\svard[\sobs{\svarb_1}{(\stimes{\svarb_2}{\svarb_3})}]}{\svarf}}{\cnlgseqp{\svard[\stimes{(\sobs{\svarb_1}{\svarb_2})}{\svarb_3}]}{\svarf}}$ & &
$\infer[\grisha{I}{2b}]{\cnlgseqp{\svard[\sos{(\stimes{\svarb_1}{\svarb_2})}{\svarb_3}]}{\svarf}}{\cnlgseqp{\svard[\sos{\svarb_1}{(\sobs{\svarb_2}{\svarb_3})}]}{\svarf}}$ \\ \\
$\infer[\grishc{I}{a}]{\cnlgseqp{\svard[\sobs{\svarb_2}{(\stimes{\svarb_1}{\svarb_3})}]}{\svarf}}{\cnlgseqp{\svard[\stimes{\svarb_1}{(\sobs{\svarb_2}{\svarb_3})}]}{\svarf}}$ & &
$\infer[\grishc{I}{b}]{\cnlgseqp{\svard[\sobs{\svarb_2}{(\stimes{\svarb_3}{\svarb_1})}]}{\svarf}}{\cnlgseqp{\svard[\sobs{(\sobs{\svarb_1}{\svarb_2})}{\svarb_3}]}{\svarf}}$ \\ \\
$\infer[\grishc{I}{c}]{\cnlgseqp{\svard[\sos{(\stimes{\svarb_1}{\svarb_3})}{\svarb_2}]}{\svarf}}{\cnlgseqp{\svard[\stimes{(\sos{\svarb_1}{\svarb_2})}{\svarb_3}]}{\svarf}}$ & &
$\infer[\grishc{I}{d}]{\cnlgseqp{\svard[\sos{(\stimes{\svarb_3}{\svarb_1})}{\svarb_2}]}{\svarf}}{\cnlgseqp{\svard[\sos{\svarb_1}{(\sos{\svarb_2}{\svarb_3})}]}{\svarf}}$ 
\end{tabular}
\end{center}
\caption{Structural rules for strong focalization, applied during context-splitting. Compared to Figure \ref{neutraldispl_grishin}: each mixed associativity principle  is split into two rules, while each mixed commutativity principle splits into four. (REVISE)}
\label{cnlg_strongfoc_struc}
\end{figure}

\noindent The current treatment of structural rules emphasizes their contribution to context splitting. In particular, whereas for the base logic the latter process is easily seen to be deterministic, the same cannot be said of the structural extensions. Compare this situation to those of logics less resource sensitive, where the non-determinism of context splitting is left implicit in the representation of sequents using lists or (multi)sets of formulas.

\begin{rem}
By restricting to structures containing no non-atomic positive formulas other than $\ntop N$, focused proof search may proceed from a non-singleton set of initial goal presentations. In particular, we will prove in $\S$4 (cf. C.\ref{final_main_theorem}) that if $\cnlgseq{\svara}{\svarb}$ according to F.\ref{cnlg_weakfoc}, then also $\cnlgseq{\svard}{\svare}$ for all $\svard\in\invp{\svara^{\inp}}$ and $\svare\in\invp{\svarb^{\inp}}$, and vice versa.
\end{rem}

\begin{vb}
Figures \ref{cnlg_strongfoc_exa} and \ref{cnlg_strongfoc_exb} revisit the derivations of F.\ref{example_deriv} from the point of view of focalization, applying the decorations of D.\ref{shift_decoration}. Note that only a single focused counterpart remains for $\cnlgseq{\stimes{(\stimes{p/q}{q})}{p\bs r}}{\bar{r}}$.
\end{vb}

\begin{figure}[!]
\begin{center}
\begin{tabular}{ccc}
$\infer[D]{\cnlgseq{\stimes{(\stimes{\pton p/q}{q})}{p\bs\pton r}}{\bar{r}}}{
\infer[\sobs{}{}]{\cnlgseqp{\sobs{q}{(\sobs{p\bs\pton r}{\bar{r}})}}{\sobs{q}{\bar{p}}}}{
\infer[I]{\cnlgseqp{q}{q}}{} & 
\infer[R]{\cnlgseqp{\sobs{p\bs\pton r}{\bar{r}}}{\bar{p}}}{
\infer[D]{\cnlgseq{\sobs{p\bs\pton r}{\bar{r}}}{p}}{
\infer[\sos{}{}]{\cnlgseqp{\sos{\bar{r}}{p}}{\sos{\bar{r}}{p}}}{
\infer[I]{\cnlgseqp{p}{p}}{} &
\infer[R]{\cnlgseqp{\bar{r}}{\bar{r}}}{
\infer[D]{\cnlgseq{\bar{r}}{r}}{
\infer[I]{\cnlgseqp{r}{r}}{}}}}}}}}$ & &

$\infer[D]{\cnlgseq{\stimes{(\stimes{\pton p/q}{q})}{p\bs\pton r}}{\bar{r}}}{
\infer[\sobs{}{}]{\cnlgseqp{\sos{\bar{r}}{(\stimes{\pton p/q}{q})}}{\sos{\bar{r}}{p}}}{
\infer[\times]{\cnlgseqp{\stimes{\pton p/q}{q}}{p}}{} &
\infer[R]{\cnlgseqp{\bar{r}}{\bar{r}}}{
\infer[D]{\cnlgseq{\bar{r}}{r}}{
\infer[I]{\cnlgseqp{r}{r}}{}}}}}$
\end{tabular}
\end{center}
\caption{Deriving $\cnlgseq{\stimes{(\stimes{\pton p/q}{q})}{p\bs\pton r}}{\bar{r}}$. Of the two derivations for the unpolarized $\cnlgseq{\stimes{(\stimes{p/q}{q})}{p\bs r}}{\bar{r}}$ in Figure \ref{example_deriv}, only one is preserved.}
\label{cnlg_strongfoc_exa}
\end{figure}

\begin{figure}[h]
\begin{center}
\begin{tabular}{c}\small
$\infer[D]{\cnlgseq{\stimes{\pton p/\ntop(q\bs\pton p)}{\ntop(\pton p/\ntop(q\bs\pton p))\bs\pton p}}{\bar{p}}}{
\infer[\sobs{}{}]{\cnlgseqp{\sobs{\ntop(\pton p/\ntop(q\bs\pton p))\bs\pton p}{\bar{p}}}{\sobs{q\bs\pton p}{\bar{p}}}}{
\infer[R]{\cnlgseqp{\ntop(\pton p/\ntop(q\bs\pton p))\bs\pton p}{q\bs\pton p}}{
\infer[D]{\cnlgseq{\ntop(\pton p/\ntop(q\bs\pton p))\bs\pton p}{\sos{\bar{p}}{q}}}{
\infer[\sos{}{}]{\cnlgseqp{\sos{\bar{p}}{q}}{\sos{\bar{p}}{\pton p/\ntop(q\bs\pton p)}}}{
\infer[R]{\cnlgseqp{q}{\pton p/\ntop(q\bs\pton p)}}{
\infer[D]{\cnlgseq{q}{\sobs{q\bs\pton p}{\bar{p}}}}{
\infer[\sos{}{}]{\cnlgseqp{\sos{\bar{p}}{q}}{\sos{\bar{p}}{q}}}{
\infer[I]{\cnlgseqp{q}{q}}{} &
\infer[R]{\cnlgseqp{\bar{p}}{\bar{p}}}{
\infer[D]{\cnlgseq{\bar{p}}{p}}{
\infer[I]{\cnlgseqp{p}{p}}{}}
}}}} &
\infer[R]{\cnlgseqp{\bar{p}}{\bar{p}}}{
\infer[D]{\cnlgseq{\bar{p}}{p}}{
\infer[I]{\cnlgseqp{p}{p}}{}}
}}}} &
\infer[R]{\cnlgseqp{\bar{p}}{\bar{p}}}{
\infer[D]{\cnlgseq{\bar{p}}{p}}{
\infer[I]{\cnlgseqp{p}{p}}{}}}}}$ \ \ \ 

$\infer[D]{\cnlgseq{\stimes{\pton p/\ntop(q\bs\pton p)}{\ntop(\pton p/\ntop(q\bs\pton p))\bs\pton p}}{\bar{p}}}{
\infer[\sos{}{}]{\cnlgseqp{\sos{\bar{p}}{\pton p/\ntop(q\bs\pton p)}}{\sos{\bar{p}}{\pton p/\ntop(q\bs\pton p)}}}{
\infer[R]{\cnlgseqp{\pton p/\ntop(q\bs\pton p))}{\pton p/\ntop(q\bs\pton p))}}{
\infer[D]{\cnlgseq{\pton p/\ntop(q\bs\pton p))}{\sobs{q\bs\pton p}{\bar{p}}}}{
\infer[\sobs{}{}]{\cnlgseqp{\sobs{q\bs\pton p}{\bar{p}}}{\sobs{q\bs\pton p}{\bar{p}}}}{
\infer[R]{\cnlgseqp{q\bs\pton p}{q\bs\pton p}}{
\infer[D]{\cnlgseq{q\bs\pton p}{\sos{\bar{p}}{q}}}{
\infer[\sos{}{}]{\cnlgseqp{\sos{\bar{p}}{q}}{\sos{\bar{p}}{q}}}{
\infer[I]{\cnlgseqp{q}{q}}{} &
\infer[I]{\cnlgseqp{\bar{p}}{\bar{p}}}{
\infer[D]{\cnlgseq{\bar{p}}{p}}{
\infer[I]{\cnlgseqp{p}{p}}{}}}}}} & 
\infer[I]{\cnlgseqp{\bar{p}}{\bar{p}}}{
\infer[D]{\cnlgseq{\bar{p}}{p}}{
\infer[I]{\cnlgseqp{p}{p}}{}}}}}} &
\infer[I]{\cnlgseqp{\bar{p}}{\bar{p}}}{
\infer[D]{\cnlgseq{\bar{p}}{p}}{
\infer[I]{\cnlgseqp{p}{p}}{}}}}}$
\end{tabular}
\end{center}
\caption{Deriving $\cnlgseq{\stimes{\pton p/\ntop(q\bs\pton p)}{\ntop(\pton p/\ntop(q\bs\pton p))\bs\pton p}}{\bar{p}}$. Both derivations of the unpolarized $\cnlgseq{\sobs{(p/(q\bs p))\bs p}{\bar{p}}}{p/(q\bs p)}$ from Figure \ref{example_deriv} are preserved.}
\label{cnlg_strongfoc_exb}
\end{figure}

\noindent We ensure closure under the display postulates and Linear Distributivity (F.\ref{neutraldispl_grishin}). 

\begin{lem}\label{closure_displ} For any $\svard,\svare,\svarf$, we have the following implications:
\begin{center}
\begin{tabular}{ll}
(a) & $\cnlgseq{\svare}{\svard}$ implies $\cnlgseq{\svard}{\svare}$ \\
(b) & $\cnlgseq{\stimes{\svard}{\svare}}{\svarf}$ implies $\cnlgseq{\svard}{\sobs{\svare}{\svarf}}$ \\
(c) & $\cnlgseq{\svard}{\stimes{\svare}{\svarf}}$ implies $\cnlgseq{\sos{\svard}{\svare}}{\svarf}$
\end{tabular}
\end{center}
\end{lem}
\begin{proof}
As a typical case, we check (b). Evidently, any derivation of $\cnlgseq{\stimes{\svard}{\svare}}{\svarf}$ must end with an application of $(D)$. The main formula $N$ must occur in either $\svard$, $\svare$ or $\svarf$. Without loss of generality, assume $N$ is in $\svard$, i.e., $\svard=\svard'[N]$:
\begin{center}
$\infer[D]{\cnlgseq{\stimes{\svard'[N]}{\svare}}{\svarf}}{
\cnlgseqp{\displ{(\stimes{\svard'[]}{\svare})}{\svarf}}{\svarf'}}$
\end{center}
for some $\svarf'\in\invp{N^{\bot}}$. Since, by definition, $\displ{(\stimes{\svard'[]}{\svare})}{\svarf}=\displ{\svard'[]}{(\sobs{\svare}{\svarf})}$, we can also derive $\cnlgseq{\svard[N]}{\sobs{\svare}{\svarf}}$:
\begin{center}
$\infer[D]{\cnlgseq{\svard[N]}{\sobs{\svare}{\svarf}}}{
\cnlgseqp{\displ{\svard'[]}{(\sobs{\svare}{\svarf})}}{\svarf'}}$
\end{center}
\end{proof}

\begin{lem}\label{closure_grish}
We have the following admissible rules (compare with Figure \ref{neutraldispl_grishin}):
\begin{center}
\begin{tabular}{ll}
(a) & In the presence of $\grisha{I}{1a,b}$, $\cnlgseq{\stimes{\svara_1}{\svara_2}}{\stimes{\svarb_2}{\svarb_1}}$ if $\cnlgseq{\sos{\svara_2}{\svarb_2}}{\sos{\svarb_1}{\svara_1}}$ \\
(b) & In the presence of $\grisha{I}{2a,b}$, $\cnlgseq{\stimes{\svara_1}{\svara_2}}{\stimes{\svarb_2}{\svarb_1}}$ if $\cnlgseq{\sobs{\svarb_1}{\svara_1}}{\sobs{\svara_2}{\svarb_2}}$ \\
(c) & In the presence of $\grishc{I}{a-d}$, $\cnlgseq{\stimes{\svara_1}{\svara_2}}{\stimes{\svarb_2}{\svarb_1}}$ if $\cnlgseq{\sos{\svarb_2}{\svara_1}}{\sobs{\svarb_1}{\svara_2}}$ 
\end{tabular}
\end{center}
\end{lem}
\begin{proof}
We demonstrate (a), the same technique applying for proving (b)-(f). Again, $\cnlgseq{\sos{\svard_2}{\svare_2}}{\sos{\svare_1}{\svard_1}}$ can only have been witnessed by a derivation ending with an application of $(D)$, so we consider four subcases, depending on whether the main formula $N$ is in $\svard_1$, $\svard_2$, $\svare_1$ or $\svare_3$:
\begin{center}
\begin{tabular}{cc}
$\infer[D]{\cnlgseq{\sos{\svard_2}{\svare_2}}{\sos{\svare_1}{\svard_1'[N]}}}{\cnlgseqp{\displ{\svard_1'[]}{(\sobs{(\sos{\svard_2}{\svare_2})}{\svare_1})}}{\svarf}}$ & 
$\infer[D]{\cnlgseq{\sos{\svard_2'[N]}{\svare_2}}{\sos{\svare_1}{\svard_1}}}{\cnlgseqp{\displ{\svard_2'[]}{(\stimes{\svare_2}{(\sos{\svare_1}{\svard_1})})}}{\svarf}}$ \\ \\
($N$ in $\svard_1$) & ($N$ in $\svard_2$) \\ \\
$\infer[D]{\cnlgseq{\sos{\svard_2}{\svare_2}}{\sos{\svare_1'[N]}{\svard_1}}}{\cnlgseqp{\displ{\svare_1'[]}{(\stimes{\svard_1}{(\sos{\svard_2}{\svare_2})})}}{\svarf}}$ &
$\infer[D]{\cnlgseq{\sos{\svard_2}{\svare_2'[N]}}{\sos{\svare_1}{\svard_1}}}{\cnlgseqp{\displ{\svare_2'[]}{(\sobs{(\sos{\svare_1}{\svard_1})}{\svard_2})}}{\svarf}}$ \\ \\
($N$ in $\svare_1$) & ($N$ in $\svare_2$)
\end{tabular}
\end{center}
for some $\svarf\in\invp{N^{\bot}}$. We receive the desired results by applications of $(\grisha{I}{1a/1b})$:
\begin{center}
\begin{tabular}{cc}
$\infer[D]{\cnlgseq{\stimes{\svard_1'[N]}{\svard_2}}{\stimes{\svare_2}{\svare_1}}}{
\infer[\grisha{I}{1b}]{\cnlgseqp{\displ{\svard_1'[]}{(\sobs{\svard_2}{(\stimes{\svare_2}{\svare_1})})}}{\svarf}}{\cnlgseqp{\displ{\svard_1'[]}{(\sobs{(\sos{\svard_2}{\svare_2})}{\svare_1})}}{\svarf}}}$ &
$\infer[D]{\cnlgseq{\stimes{\svard_1}{\svard_2'[N]}}{\stimes{\svare_2}{\svare_1}}}{
\infer[\grisha{I}{1a}]{\cnlgseqp{\displ{\svard_2'[]}{(\sos{(\stimes{\svare_2}{\svare_1})}{\svard_1})}}{\svarf}}{
\cnlgseqp{\displ{\svard_2'[]}{(\stimes{\svare_2}{(\sos{\svare_1}{\svard_1})})}}{\svarf}}}$ \\ \\
$\infer[D]{\cnlgseq{\stimes{\svard_1}{\svard_2}}{\stimes{\svare_2}{\svare_1'[N]}}}{
\infer[\grisha{I}{1a}]{\cnlgseqp{\displ{\svare_1'[]}{(\sos{(\stimes{\svard_1}{\svard_2})}{\svare_2})}}{\svarf}}{
\cnlgseqp{\displ{\svare_1'[]}{(\stimes{\svard_1}{(\sos{\svard_2}{\svare_2})})}}{\svarf}}}$ & 
$\infer[D]{\cnlgseq{\stimes{\svard_1}{\svard_2}}{\stimes{\svare_2'[N]}{\svare_1}}}{
\infer[\grisha{I}{1b}]{\cnlgseqp{\displ{\svare_2'[]}{(\sobs{\svare_1}{(\stimes{\svard_1}{\svard_2})})}}{\svarf}}{
\cnlgseqp{\displ{\svare_2'[]}{(\sobs{(\sos{\svare_1}{\svard_1})}{\svard_2})}}{\svarf}}}$
\end{tabular}
\end{center}
\end{proof}

\noindent We proceed to demonstrate soundness of strong focalization w.r.t. derivability in \calclg/\calccnl. Combined with T.\ref{unpol_weak}, all that will be left to explicate in order to close the square of F.\ref{summ_results} is the correspondence between weak and strong focalization, to which we will dedicate the entirety of $\S$4.

\begin{defs}
We define the forgetful maps taking a positive $P$ or negative $N$ of \lgpol \ (\cnlpol) into shift-free formulas $\forgetp{P}$ and $\forgetn{N}$ of \calclg \ (\calccnl):
\begin{center}
\begin{tabular}{rclcrcl}
$\forgetp{p}$ & $\defeq$ & $p$ & & $\forgetn{\bar{p}}$ & $\defeq$ & $\bar{p}$ \\
$\forgetp{P\lgtimes Q}$ & $\defeq$ & $\forgetp{P}\lgtimes\forgetp{Q}$ & & $\forgetn{M\lgplus N}$ & $\defeq$ & $\forgetn{M}\lgplus\forgetn{N}$ \\
$\forgetp{N\obs P}$ & $\defeq$ & $\forgetp{N}\obs\forgetp{P}$ & & $\forgetn{M/Q}$ & $\defeq$ & $\forgetn{M}/\forgetn{Q}$ \\
$\forgetp{P\os N}$ & $\defeq$ & $\forgetp{P}\os\forgetp{N}$ & & $\forgetn{Q\bs M}$ & $\defeq$ & $\forgetn{Q}\bs\forgetn{M}$ \\
$\forgetp{P\lor Q}$ & $\defeq$ & $\forgetp{P}\lor\forgetp{Q}$ & & $\forgetn{M\land N}$ & $\defeq$ & $\forgetn{M}\land\forgetn{N}$ \\
$\forgetp{\ntop N}$ & $\defeq$ & $\forgetn{N}$ & & $\forgetn{\pton P}$ & $\defeq$ & $\forgetp{P}$
\end{tabular}
\end{center}
At the level of structures, $\forgets{\svara}$ denotes the result of substituting occurrences of $N$ by $\forgetn{N}$, while leaving positive atoms intact. Finally, for presentations, $\forgets{\svard,\svare}\defeq\pres{\forgets{\svard}}{\forgets{\svare}}$.
\end{defs}

\noindent Our goal is demonstrate $\cnlgseq{\svard}{\svare}$ implies $\cnlgseq{\forgets{\svard}}{\forgets{\svare}}$.

\begin{lem}\label{soundness_displ}
For any $\svard[]$, $\svare$ and $\svarf$, $\cnlgseq{\forgets{\svard[\svarf]}}{\forgets{\svare}}$ iff $\cnlgseq{\forgets{\displ{\svard[]}{\svare}}}{\forgets{\svarf}}$.
\end{lem}
\begin{proof}
By induction on $\svard[]$. The base case is immediate from $(dp)$. For the inductive cases, we check $\svard[]=\stimes{\svard_1[]}{\svard_2}$ and $\svard[]=\stimes{\svard_1}{\svard_2[]}$:
\begin{center}
\begin{tabular}{ccc}
$\infer=[Dp]{\cnlgseq{\stimes{\forgets{\svard_1[\svarf]}}{\forgets{\svard_2}}}{\forgets{\svare}}}{
\infer=[IH]{\cnlgseq{\forgets{\svard_1[\svarf]}}{\sobs{\forgets{\svard_2}}{\forgets{\svare}}}}{
\cnlgseq{\forgets{\displ{\svard_1[]}{(\sobs{\svard_2}{\svare})}}}{\forgets{\svarf}}}}$ & &
$\infer=[Dp]{\cnlgseq{\sobs{\forgets{\svard_1[\svarf]}}{\forgets{\svard_2}}}{\forgets{\svare}}}{
\infer=[IH]{\cnlgseq{\forgets{\svard_1[\svarc]}}{\sos{\forgets{\svard_2}}{\forgets{\svarf}}}}{
\cnlgseq{\forgets{\displ{\svard_1[]}{(\sos{\svard_2}{\svare})}}}{\forgets{\svarf}}}}$
\end{tabular}
\end{center}
the desired result being immediate from $\displ{(\stimes{\svard_1[]}{\svard_2})}{\svare}=\displ{\svard_1[]}{(\sobs{\svard_2}{\svare})}$ and $\displ{(\sobs{\svard_1[]}{\svard_2})}{\svare}=\displ{\svard_1[]}{(\sos{\svard_2}{\svare})}$. 
In applying the induction hypothesis, we implicitly assumed $\forgets{\displ{\svard[]}{\forgets{\svare}}}=\forgets{\displ{\svard[]}{\svare}}$, which is easy to check.
\end{proof}
\begin{cor}\label{soundness_displ_eq} For any $\omega[]$ and $\svard$, $\forgets{\omega[\svard]}\vdash$ iff $\cnlgseq{\forgets{\omega^*[]}}{\forgets{\svard}}$.
\end{cor}

\begin{lem}\label{soundness_pos_inf}
For any $P$, $\svara$, if $\cnlgseq{\svara}{\forgets{\svare}}$ for all $\svare\in\invp{P}$, then $\cnlgseq{\svara}{\forgetp{P}}$. I.e., the following inference is admissible for (unpolarized) \calclg$_{\emptyset}$, and hence \calclg$_I$/\calccnl:
\begin{center}
$\infer{\cnlgseq{\svara}{\forgetp{P}}}{\{\cnlgseq{\svara}{\forgets{\svare}}\}_{\svare\in\invp{P}}}$
\end{center}
\end{lem}
\begin{proof}
By induction on $P$. If $P=p$ or $P=\ntop N$, $\invp{P}=P$ and the desired result is immediate. For the remaining inductive cases, we consider explicitly $P=P_1\lor P_2$ and $P=P_1\lgtimes P_2$. The former is demonstrated thus:
\begin{center}
$\infer[\lor]{\cnlgseq{\svara}{\forgetp{P_1}\lor\forgetp{P_2}}}{
\infer[IH]{\cnlgseq{\svara}{\forgetp{P_1}}}{\{\cnlgseq{\svara}{\svard}\}_{\svard\in\invp{P_1}}} &
\infer[IH]{\cnlgseq{\svara}{\forgetp{P_2}}}{\{\cnlgseq{\svara}{\svare}\}_{\svare\in\invp{P_2}}}}$
\end{center}
noting $\invp{P_1\lor P_2}=\invp{P_1}\bigcup\invp{P_2}$. In case $P=P_1\lgtimes P_2$, we have
\begin{center}
$\infer[\lgtimes]{\cnlgseq{\svara}{\forgetp{P_1}\lgtimes\forgetp{P_2}}}{
\infer=[Dp]{\cnlgseq{\svara}{\stimes{\forgetp{P_1}}{\forgetp{P_2}}}}{
\infer[IH]{\cnlgseq{\sos{\svara}{\forgetp{P_1}}}{\forgetp{P_2}}}{
\svare\in\invp{P_2} &
\infer=[Dp]{\cnlgseq{\sos{\svara}{\forgetp{P_1}}}{\forgets{\svare}}}{
\infer[IH]{\cnlgseq{\sobs{\svare}{\svara}}{\forgetp{P_1}}}{
\svard\in\invp{P_1} &
\infer=[Dp]{\cnlgseq{\sobs{\svare}{\svara}}{\svard}}{\cnlgseq{\svara}{\stimes{\svard}{\svare}}}}}}}}$
\end{center}
noting $\invp{P_1\lgtimes P_2}=\{\stimes{\svard}{\svare} \ | \ \svard\in\invp{P_1}, \ \svare\in\invp{P_2}\}$.
\end{proof}

\begin{lem}\label{soundness_neg_inf}
For any $N$ and $\svare\in\invp{N^{\bot}}$, $\cnlgseq{\forgetn{N}}{{\forgetp{\svare^+}}}$.
\end{lem}
\begin{proof}
By induction on $N$. If $N=\pton P$ or $N=\bar{p}$, the desired result is immediate by applying $(I)$. For the remaining inductive cases, we check $N=N_1\land N_2$ and $N=N_1\lgplus N_2$. In the former case, note $\svare\in\invp{N_2^{\bot}\lor N_1^{\bot}}$ iff $\svare\in\invp{N_2^{\bot}}$ or $\svare\in\invp{N_1^{\bot}}$. Thus, applying the induction hypotheses, we have
\begin{center}
\begin{tabular}{ccc}
$\infer=[dp]{\cnlgseq{\forgetn{N_1}\land\forgetn{N_2}}{\forgetp{\svare^+}}}{
\infer[\land^l]{\cnlgseq{\forgetp{\svare^+}}{\forgetn{N_1}\land\forgetn{N_2}}}{
\infer=[dp]{\cnlgseq{\forgetp{\svare^+}}{\forgetn{N_1}}}{
\infer[IH]{\cnlgseq{\forgetn{N_1}}{\forgetp{\svare^+}}}{}}}}$ & &
$\infer=[dp]{\cnlgseq{\forgetp{\svare^+}}{\forgetn{N_1}\land\forgetn{N_2}}}{
\infer[\land^r]{\cnlgseq{\forgetp{\svare^+}}{\forgetn{N_1}\land\forgetn{N_2}}}{
\infer=[dp]{\cnlgseq{\forgetn{N_2}}{\forgetp{\svare^+}}}{
\infer[IH]{\cnlgseq{\forgetp{\svare^+}}{\forgetn{N_2}}}{}}}}$
\end{tabular}
\end{center}
In case $N=N_1\lgplus N_2$, note $\svare\in\invp{N_2^{\bot}\lgtimes N_1^{\bot}}$ iff $\svare=\stimes{\svare_2}{\svare_1}$ for $\svare_1\in\invp{N_1^{\bot}}$ and $\svare_2\in\invp{N_2^{\bot}}$. Hence, by the induction hypotheses,
\begin{center}
$\infer[\lgtimes]{\cnlgseq{\forgetn{N_1}\lgplus\forgetn{N_2}}{\forgetp{\svare_2^+}\lgtimes\forgetp{\svare_1^+}}}{
\infer=[dp]{\cnlgseq{\forgetn{N_1}\lgplus\forgetn{N_2}}{\stimes{\forgetp{\svare_2^+}}{\forgetp{\svare_1^+}}}}{
\infer[\lgplus]{\cnlgseq{\stimes{\forgetp{\svare_2^+}}{\forgetp{\svare_1^+}}}{\forgetn{N_1}\lgplus\forgetn{N_2}}}{
\infer=[dp]{\cnlgseq{\forgetp{\svare_1^+}}{\forgetn{N_1}}}{
\infer[IH]{\cnlgseq{\forgetn{N_1}}{\forgetp{\svare_1^+}}}{}} &
\infer=[dp]{\cnlgseq{\forgetp{\svare_2^+}}{\forgetn{N_2}}}{
\infer[IH]{\cnlgseq{\forgetn{N_2}}{\forgetp{\svare_2^+}}}{}}}}}$
\end{center}
\end{proof}

\begin{lem}\label{strong_unpol_lemma}
We have the following implications:
\begin{center}
\begin{tabular}{rcl}
$\cnlgseq{\svard}{\svare}$ & $\Longrightarrow$ & $\cnlgseq{\forgets{\svard}}{\forgets{\svare}}$ \\
$\cnlgseqp{\svard}{\svare}$ & $\Longrightarrow$ & $\cnlgseq{\forgets{\svard}}{{\forgetn{\svare^-}}}$
\end{tabular}
\end{center}
\end{lem}
\begin{proof}
By a mutual induction. The case $(I)$ is immediate, so we are left to check \\

\noindent\textbf{Case} $(\stimes{}{})$, $(\sobs{}{})$, $(\sos{}{})$. We check $(\stimes{}{})$, the others being similar. By induction hypothesis, $\cnlgseq{\forgets{\svard}}{\forgetn{\svard'^-}}$ and $\cnlgseq{\forgets{\svare}}{\forgetn{\svare'^-}}$, so we apply $(\lgplus)$:
\begin{center}
$\infer{\cnlgseq{\stimes{\forgets{\svard}}{\forgets{\svare}}}{\forgetn{\svard'^-}\lgplus\forgetn{\svare'^-}}}{
\infer[IH]{\cnlgseq{\forgets{\svard}}{\forgetn{\svard'^-}}}{} &
\infer[IH]{\cnlgseq{\forgets{\svare}}{\forgetn{\svare'^-}}}{}}$
\end{center}

\noindent\textbf{Case} 
$(\grisha{I}{1a/1b})$, $(\grisha{I}{2a/2b})$, $(\grishc{I}{a-d})$, $(\cnlos)$, $(\cnlobs)$. We check $(\grisha{I}{1a})$. By the induction hypothesis, $\cnlgseq{\forgets{\svard[\stimes{\svare_1}{(\sos{\svare_2}{\svare_3})}]}}{\forgetn{\svarf^-}}$. We might try using C.\ref{soundness_displ_eq}, but this necessitates writing $\displ{\svard[]}{\svarf^-}$, which need not be defined w.r.t. $X^{\tau}$, seeing as we cannot assume $\svarf^-\in X^{\tau}$. Thus, we prove as an additional Lemma, for arbitrary $\svara$ and proceeding by induction on $\svard[]$, admissibility of 
\begin{center}
$\infer{\cnlgseq{\forgets{\svard[\sos{(\stimes{\svare_1}{\svare_2})}{\svare_3}]}}{\svara}}{\cnlgseq{\forgets{\svard[\stimes{\svare_1}{(\sos{\svare_2}{\svare_3})}]}}{\svara}}$
\end{center}
instantiating $\svara$ with $\forgetn{\svarf^-}$ for the desired result. In the base case $\svard[]=[]$,
\begin{center}
$\infer=[dp]{\cnlgseq{\sos{(\stimes{\forgets{\svare_1}}{\forgets{\svare_2}})}{\forgets{\svare_3}}}{\svara}}{
\infer[\grisha{I}{1}]{\cnlgseq{\stimes{\forgets{\svare_1}}{\forgets{\svare_2}}}{\stimes{\forgets{\svare_3}}{\svara}}}{
\infer=[Dp]{\cnlgseq{\sos{\forgets{\svare_2}}{\forgets{\svare_3}}}{\sos{\svara}{\forgets{\svare_1}}}}{
\cnlgseq{\stimes{\forgets{\svare_1}}{(\sos{\forgets{\svare_2}}{\forgets{\svare_3}})}}{\svara}}}}$
\end{center}
Next, we check $\svard[]=\stimes{\svard'[]}{\svard''}$, the other inductive cases being handled similarly. 
\begin{center}
$\infer=[dp]{\cnlgseq{\stimes{\forgets{\svard'[\sos{(\stimes{\svare_1}{\svare_2})}{\svare_3}]}}{\forgets{\svard''}}}{\svara}}{
\infer[IH]{\cnlgseq{\forgets{\svard'[\sos{(\stimes{\svare_1}{\svare_2})}{\svare_3}]}}{\sobs{\forgets{\svard''}}{\svara}}}{
\infer=[dp]{\cnlgseq{\forgets{\svard'[\stimes{\svare_1}{(\sos{\svare_2}{\svare_3})}]}}{\sobs{\forgets{\svard''}}{\svara}}}{
\cnlgseq{\stimes{\forgets{\svard'[\stimes{\svare_1}{(\sos{\svare_2}{\svare_3})}]}}{\forgets{\svard''}}}{\svara}}}}$
\end{center}

\noindent\textbf{Case} $(R)$. Immediate by L.\ref{soundness_pos_inf}. \\

\noindent\textbf{Case} $(D)$. By induction hypothesis, $\cnlgseq{\forgets{\omega^*[]}}{\forgetn{\svare^-}}$, while $\cnlgseq{N}{\forgetp{\svare^+}}$ by L.\ref{soundness_neg_inf}. An easy induction will conform $\forgetn{\svare^-}$ and $\forgetp{\svare^+}$ are dual, and hence we can apply $(T)$, invoking C.\ref{soundness_displ_eq} afterwards:
\begin{center}
$\infer[C.\ref{soundness_displ_eq}]{\forgets{\omega[N]}\vdash}{
\infer[T]{\cnlgseq{\forgets{\omega^*[]}}{\forgetn{N}}}{
\infer[L.\ref{soundness_neg_inf}]{\cnlgseq{\forgetn{N}}{\forgetp{\svare^+}}}{} &
\infer[IH]{\cnlgseq{\forgets{\omega^*[]}}{\forgetn{\svare^-}}}{}}}$
\end{center}
\end{proof}

\begin{thm}\label{strong_unpol}
$\cnlgseq{\svara}{\svarb}$ in \calclg/\calccnl \ only if $\cnlgseq{\svard}{\svare}$ for all $\svard\in\invp{\svara^+}$, $\svare\in\invp{\svarb^+}$.

\end{thm}

\begin{proof}
Assume (*) $\cnlgseq{\svarc}{\svarc'}$ if $(\forall\svard\in\invp{\svarc^{\inp}})(\cnlgseq{\forgets{\svard}}{\svarc'})$ in \calclg \ (\calccnl).
Then 
\begin{center}
$\infer[*]{\cnlgseq{\svara}{\svarb}}{
(\forall\svard\in\invp{\svara^{\inp}}) & 
\infer=[Dp]{\cnlgseq{\forgets{\svard}}{\svarb}}{
\infer[*]{\cnlgseq{\svarb}{\forgets{\svard}}}{
(\forall\svare\in\invp{\svarb^{\inp}}) &
\infer=[Dp]{\cnlgseq{\forgets{\svare}}{\forgets{\svard}}}{
\infer[L.\ref{strong_unpol_lemma}]{\cnlgseq{\forgets{\svard}}{\forgets{\svare}}}{}}}}}$
\end{center}
Suffice it to show the admissibility of (*). In the base case, $\svarc=A$ and we apply Lemma \ref{soundness_pos_inf} with $\semt{A}$ if $\semt{A}$ is positive, and with $\ntop\semt{A}$ otherwise. For the inductive cases, consider $\Theta=\stimes{\Theta_1}{\Theta_2}$, handled thus:
\begin{center}
$\infer=[Dp]{\cnlgseq{\stimes{\svarc_1}{\svarc_2}}{\svarc'}}{
\infer[IH]{\cnlgseq{\svarc_1}{\sobs{\svarc_2}{\svarc'}}}{
(\forall\svard\in\invp{\svara^{\inp}}) &
\infer=[Dp]{\cnlgseq{\svard}{\sobs{\svarc_2}{\svarc'}}}{
\infer[IH]{\cnlgseq{\sos{\svarc'}{\svard}}{\svarc_2}}{
(\forall\svare\in\invp{\svarc_2}) &
\infer=[Dp]{\cnlgseq{\sos{\svarc'}{\svard}}{\svare}}{
\cnlgseq{\stimes{\svard}{\svare}}{\svarc'}}}}}}$
\end{center}
\end{proof}

\section{Normalization as Completeness}
The current section demonstrates provability in \lgpol \ (\cnlpol) \ (F.\ref{cnlg_weakfoc}) implies focused provability (F.\ref{cnlg_strongfoc}). Note the converse direction already obtains by composing Theorems \ref{strong_unpol} and \ref{unpol_weak}. The standard approach proceeds via Cut elimination, as explained in \cite{laurent04}. Here, instead, we provide a model-theoretic argument along the lines of \cite{okada02} and \cite{herbelinlee09}. That is, we define  phase models for \textbf{LG} and \textbf{CNL} and construct a syntactic model for which we show `truth' to imply focused provability. Composed with soundness for the derivations of F.\ref{cnlg_weakfoc}, the desired result immediately follows. We define our phase models and establish soundness in $\S$4.1, while $\S$4.2 will be dedicated to showing completeness.

\subsection{Phase models}
\begin{defs}\label{phase_space}
A \textit{phase space} is a 5-tuple $\langle P,\stimes{}{},\sos{}{},\sobs{}{},\bot\rangle$ where:
\begin{enumerate}
\item $P$ is a non-empty set of \textit{phases} with operations $\stimes{}{},\sos{}{},\sobs{}{}:P\times P\rightarrow P$. We use metavariables $x,y,z$ for denoting elements of $P$ and $A,B,C$ for denoting subsets of $P$.
\item $\bot\subseteq P\times P$ s.t.
\begin{center}
\begin{tabular}{rcl}
$\langle x,y\rangle\in\bot$ & $\Rightarrow$ & $\langle y,x\rangle\in\bot$ \\
$\langle\stimes{x}{y},z\rangle\in\bot$ & $\Leftrightarrow$ & $\langle x,\sobs{y}{z}\rangle\in\bot$ \\
$\langle x,\stimes{y}{z}\rangle\in\bot$ & $\Leftrightarrow$ & $\langle\sos{x}{y},z\rangle\in\bot$
\end{tabular}
\end{center}
\item A phase space may be required to satisfy further conditions depending on which structural rules are added to the base logic, as detailed in Table \ref{constr_strucpost}. 
\end{enumerate}
As usual, we often identify a phase space by its carrier set $P$. Given a phase space, the operation $\cdot^{\bot}:\mathscr{P}(P)\rightarrow\mathscr{P}(P)$ is defined by mapping $A\subseteq P$ to $\{ x \ | \ (\forall y\in A)(\langle x,y\rangle\in\bot)\}$.
\end{defs}
\begin{table}
\begin{center}
\begin{tabular}{@{}|l|l|l|@{}}\hline
\multirow{2}{*}{\textbf{Constraints}} & \multicolumn{2}{|l|}{\textbf{Structural Postulates}} \\ 
 & & F.\ref{cnlg_strongfoc_struc} \\ \hline
$\langle\sos{y}{u},\sos{v}{x}\rangle\in\bot\Rightarrow\langle\stimes{x}{y},\stimes{u}{v}\rangle\in\bot$ & F.\ref{neutraldispl_grishin}, $\grisha{I}{1}$ & $\grisha{I}{1a}$, $\grisha{I}{1b}$ \\ 
$\langle\sobs{v}{x},\sobs{y}{u}\rangle\in\bot\Rightarrow\langle\stimes{x}{y},\stimes{u}{v}\rangle\in\bot$ & F.\ref{neutraldispl_grishin}, $\grisha{I}{2}$ & $\grisha{I}{2a}$, $\grisha{I}{2b}$ \\
$\langle\sos{u}{x},\sobs{v}{y}\rangle\in\bot\Rightarrow\langle\stimes{x}{y},\stimes{u}{v}\rangle\in\bot$ & F.\ref{neutraldispl_grishin}, $\grishc{I}{}$ & $\grishc{I}{a}$, $\grishc{I}{b}$, $\grishc{I}{c}$, $\grishc{I}{d}$ \\
$\langle\sos{x}{y},z\rangle\in\bot\Leftrightarrow\langle\stimes{x}{y},z\rangle\in\bot\Leftrightarrow\langle\sobs{x}{y},z\rangle\in\bot$ & F.\ref{neutraldispl_cnl} & $(\cnlos)$, $(\cnlobs)$ \\ \hline
\end{tabular}
\end{center}
\caption{Structural postulates and conditions on phase spaces.}
\label{constr_strucpost}
\end{table}

\begin{rem}
If we were to restrict our attention to \textbf{CNL}, a more parsimonious definition for phase spaces seems naturally available: take any 3-tuple $\langle P,\stimes{}{},\bot\rangle$, where $\stimes{}{}:P\times P\rightarrow P$ and $\bot\subseteq P\times P$ s.t., for all $x,y\in P$,
\begin{center}
\begin{tabular}{rcl}
$\langle x,y\rangle\in\bot$ & $\Rightarrow$ & $\langle y,x\rangle\in\bot$ \\
$\langle\stimes{x}{y},z\rangle\in\bot$ & $\Leftrightarrow$ & $\langle x,\stimes{y}{z}\rangle\in\bot$
\end{tabular}
\end{center}
\end{rem}
\noindent The following are some easy observations on the operation $\cdot^{\bot}$ on phase spaces.

\begin{lem}\label{galois} Given $P$, we have $A\subseteq B^{\bot}$ iff $B\subseteq A^{\bot}$ ($A,B\in\mathscr{P}(P)$). Equivalently, $A\subseteq A^{\bot\bot}$, $A\subseteq B$ implies $B^{\bot}\subseteq A^{\bot}$ and $A^{\bot\bot\bot}\subseteq A^{\bot}$. In other words, $\cdot^{\bot}$ is a \textit{Galois connection}, and hence $\cdot^{\bot\bot}$ a closure operator, meaning (at the cost of some redundancy), 
$A\subseteq A^{\bot\bot}$, $A\subseteq B$ implies $A^{\bot\bot}\subseteq B^{\bot\bot}$, $(A^{\bot\bot})^{\bot\bot}\subseteq A^{\bot\bot}$.
\end{lem}
\noindent Formulas will be interpreted by \textit{facts}: subsets $A\subseteq P$ s.t. $A=A^{\bot\bot}$. The following is a consequence of the well-known property of closure operators being closed under intersection:
\begin{lem}\label{closure_intersect}
Facts are closed under finite intersections.
\end{lem}
\begin{defs}
A \textit{model} consists of a phase space $P$ and a valuation $v$ taking positive atoms $p$ into facts. $v$ extends to maps $\valp{\cdot}$ and $\valn{\cdot}$, defined by mutual induction and acting on arbitrary positive and negative formulas respectively:
\begin{center}
\begin{tabular}{rclrcl}
$\valp{p}$ & $\defeq$ & $v(p)$ & $\valn{\bar{p}}$ & $\defeq$ & $v(p)$ \\
$\valp{P\lgtimes Q}$ & $\defeq$ & $\valp{P}\times\valp{Q}$ &
$\valn{M\lgplus N}$ & $\defeq$ & $\valn{N}\times\valn{M}$ \\
$\valp{P\os N}$ & $\defeq$ & $\valp{P}\leftarrow\valn{N}$ &
$\valn{M/Q}$ & $\defeq$ & $\valp{Q}\rightarrow\valn{M}$ \\
$\valp{N\obs P}$ & $\defeq$ & $\valn{N}\rightarrow\valp{P}$ &
$\valn{Q\bs M}$ & $\defeq$ & $\valn{M}\rightarrow\valp{Q}$ \\
$\valp{P\lor Q}$ & $\defeq$ & $\valp{P}\bigcap\valp{Q}$ &
$\valn{M\land N}$ & $\defeq$ & $\valn{M}\bigcap\valn{N}$ \\
$\valp{\ntop N}$ & $\defeq$ & $\valn{N}^{\bot}$ &
$\valn{\pton P}$ & $\defeq$ & $\valp{P}^{\bot}$
\end{tabular}
\end{center}
\noindent Here, we have employed the following operations, evidently facts by L.\ref{galois}:
\begin{center}
\begin{tabular}{rl}
$\times:\mathscr{P}(P)\times \mathscr{P}(P)\rightarrow\mathscr{P}(P)$, & $\langle A,B\rangle\mapsto\{\stimes{x}{y} \ | \ x\in A^{\bot}, \ y\in B^{\bot}\}^{\bot}$ \\ 
$\leftarrow:\mathscr{P}(P)\times \mathscr{P}(P)\rightarrow\mathscr{P}(P)$, & $\langle A,B\rangle\mapsto\{\sos{x}{y} \ | \ x\in A^{\bot}, y\in B^{\bot}\}^{\bot}$ \\
$\rightarrow:\mathscr{P}(P)\times \mathscr{P}(P)\rightarrow\mathscr{P}(P)$, & $\langle A,B\rangle\mapsto\{\sobs{x}{y} \ | \ x\in A^{\bot}, y\in B^{\bot}\}^{\bot}$
\end{tabular}
\end{center}
\end{defs}
\begin{lem}\label{shiftbot}
$\valp{P}=\valn{\pton P}^{\bot}$ and $\valn{N}=\valp{\ntop N}^{\bot}$ for any $N,P$.
\end{lem}
\begin{proof}
Immediate, since the sets involved are facts.
\end{proof}
\begin{lem}\label{valduality}
For any $N,P$, $\valp{P}=\valn{P^{\bot}}$ and (dually) $\valn{N}=\valp{N^{\bot}}$.
\end{lem}
\begin{proof}
By a straightforward inductive argument.
\end{proof}
\begin{lem}\label{foureq}
We have the following equivalences:
\begin{center}
\begin{tabular}{clcl}
 & $\valn{\pton\svara^+}\subseteq\valn{\svarb^-}$ & $\Leftrightarrow$ & $\valn{\pton\svarb^+}\subseteq\valn{\svara^-}$ \\
$\Leftrightarrow$ & $\valp{\ntop\svarb^-}\subseteq\valp{\svara^+}$ & $\Leftrightarrow$ & $\valp{\ntop\svara^-}\subseteq\valp{\svarb^+}$
\end{tabular}
\end{center}
\end{lem}
\begin{proof}
Recalling $\svarc^{+\bot}=\svarc^-$ and $\svarc^{-\bot}=\svarc^+$ for arbitrary $\svarc$, we have
\begin{center}
\begin{tabular}{lclr}
$\valp{\ntop\svarb^-}\subseteq\valp{\svara^+}$ & iff & $\valp{\svara^+}^{\bot}\subseteq\valp{\ntop\svarb^-}^{\bot}$ & (Lemma \ref{galois}) \\
 & iff & $\valn{\pton\svara^+}\subseteq\valn{\svarb^-}$ & (Lemma \ref{shiftbot})
\end{tabular}
\end{center}
and
\begin{center}
\begin{tabular}{lclr}
$\valp{\ntop\svarb^-}\subseteq\valp{\svara^+}$ \ & \ iff \ & \ $\valp{\svara^+}^{\bot}\subseteq\valp{\ntop\svarb^-}^{\bot}$ \ & \ (Lemma \ref{galois}) \\
 & \ iff \ & \ $\valn{\svara^-}^{\bot}\subseteq\valn{\pton\svarb^+}^{\bot}$ \ & \ (Lemma \ref{valduality}) \\
 & \ iff \ & \ $\valp{\ntop\svara^-}\subseteq\valp{\svarb^+}$ \ & \ (Lemma \ref{shiftbot})
\end{tabular}
\end{center}
and similarly $\valn{\pton\svara^+}\subseteq\valn{\svarb^-}$ iff $\valn{\pton\svarb^+}\subseteq\valn{\svara^-}$.
\end{proof}
\noindent We state soundness for sequent derivability in \lgpol and \cnlpol (F.\ref{cnlg_weakfoc}).
\begin{thm}\label{soundness_sem} All phase models satisfy the following implications:
\begin{center}
\begin{tabular}{rcl}
$\cnlgseq{\svara}{\svarb}$ & \ $\Longrightarrow$ \ & $\valp{\ntop\svara^-}\subseteq\valp{\svarb^+}$ \\
$\cnlgseqp{\svara}{P}$ & \ $\Longrightarrow$ \ & $\valp{P}\subseteq\valp{\svara^+}$
\end{tabular}
\end{center}
\end{thm}

\begin{proof}
By induction, freely applying L.\ref{foureq}. Note axioms $(I)$ and Cut $(T)$ trivially reduce to reflexivity and transitivity of set inclusion, while $(\ntop L)$ and $(\ntop R)$ are immediate by L.\ref{foureq}. The cases $(\lor L)$ and $(\lor R)$ are equally trivial, reducing to the defining properties of greatest lower bounds. This leaves us to check \\

\noindent\textbf{Case} $(dp)$. As a typical instance, we check $\cnlgseq{\stimes{\svara}{\svarb}}{\svarc}$ if $\cnlgseq{\svara}{\sobs{\svarb}{\svarc}}$. The following hypotheses will be used:
\begin{center}
\begin{tabular}{rl}
(IH) & $\valp{\ntop\svara^-}\subseteq\valp{(\sobs{\svarb}{\svarc})^+}$, iff $\valn{\svara^-}^{\bot}\subseteq\valn{\svarb^-}\rightarrow\valp{\svarc^+}$ \\
(a) & $x\in\valn{\svarc^-}^{\bot}=\valp{\svarc^+}^{\bot}$ \\
(b) & $y\in\valp{\svara^+}^{\bot}=\valn{\svara^-}^{\bot}$ \\
(c) & $z\in\valp{\svarb^+}^{\bot}=\valn{\svarb^-}^{\bot}$
\end{tabular}
\end{center}
(IH) being the induction hypothesis. We desire $\valp{\ntop(\stimes{\svara}{\svarb})^-}\subseteq\valp{\svarc^+}$, iff $\valp{\ntop\svarc^-}\subseteq\valp{(\stimes{\svara}{\svarb})^+}$ by L.\ref{foureq}, iff $\valn{\svarc^-}^{\bot}\subseteq\valp{\svara^+}\times\valp{\svarb^+}$ after unfolding. So assume (a)-(c). We show $\langle x,\stimes{y}{z}\rangle\in\bot$, iff $\langle\sobs{z}{x},y\rangle\in\bot$. By (b) and (IH), $y\in\{\sobs{z}{x} \ | \ z\in\valn{\svarb^-}^{\bot}, x\in\valp{\svarc^+}^{\bot}\}^{\bot}$, so we apply (a) and (c). \\

\noindent\textbf{Cases} $(\lgtimes L)$, $(\obs L)$, $(\os L)$. Immediate, upon the realization that $\valp{(\stimes{P}{Q})^+}=\valp{P\lgtimes Q}$, $\valp{(\sobs{N^{\bot}}{P})^+}=\valp{N\obs P}$ and $\valp{(\sos{P}{N^{\bot}})^+}=\valp{P\os N}$. \\

\noindent\textbf{Cases} $(\lgtimes R)$, $(\os R)$, $(\obs R)$. 
We explicitly check $(\os R)$. By the induction hypothesis, $\valp{N^{\bot}}=\valn{N}\subseteq\valp{\svarb^+}$ and $\valp{P}\subseteq\valp{\svara^+}$. Thus, by L.\ref{galois},
\begin{center}
$\{ \sos{x}{y} \ | \ x\in\valp{\svara^+}^{\bot}, y\in\valp{\svarb^+}^{\bot} \}\subseteq\{ \sos{x}{y} \ | \ x\in\valp{P}^{\bot}, y\in\valn{N}^{\bot}\}$
\end{center}
with another application of L.\ref{galois} deriving the desired $\valp{P\os N}\subseteq\valp{(\sos{\svara}{\svarb})^+}$, noting $\valp{\svara^+\os\Delta^-}=\valp{\Gamma^+}\leftarrow\valn{\Delta^-}=\valp{\Gamma^+}\leftarrow\valp{\Delta^+}$. \\

\noindent \textbf{Cases} $(\grisha{I/IV}{1/2})$, $(\grishc{I/IV}{})$, $(\cnlos)$, $(\cnlobs)$. As a typical instance, we check $(\grisha{I}{1})$, i.e., $\cnlgseq{\stimes{\svara_1}{\svara_2}}{\stimes{\svarb_2}{\svarb_1}}$ if $\cnlgseq{\sos{\svara_2}{\svarb_2}}{\sos{\svarb_1}{\svara_1}}$. We use the following hypotheses:
\begin{center}
\begin{tabular}{rl}
(F) & $\langle\sos{z}{u},\sos{v}{y}\rangle\in\bot\Rightarrow\langle\stimes{y}{z},\stimes{u}{v}\rangle\in\bot$ \\
(IH) & $\valp{\ntop(\sos{\svara_2}{\svarb_2})^-}\subseteq\valp{(\sos{\svarb_1}{\svara_1})^+}$, \\ & iff $(\valn{\svara_2^-}\leftarrow\valn{\svarb_2^-})^{\bot}\subseteq\valp{\svarb_1}\leftarrow\valp{\svara_1}$ \\
(a) & $x\in(\valn{\svara_1^-}\times\valn{\svara_2^-})^{\bot}$ \\
(b) & $y\in\valp{\svarb_2^+}^{\bot}=\valn{\svarb_2^-}^{\bot}$ \\
(c) & $z\in\valp{\svarb_1^+}^{\bot}$ \\
(d) & $u\in\valn{\svara_1^-}^{\bot}$ \\
(e) & $v\in\valn{\svara_2^-}^{\bot}=\valp{\svara_2^+}^{\bot}$
\end{tabular}
\end{center}
recalling (F) to be the frame condition associated with $\grisha{I}{1}$. We must show $\valp{\ntop(\stimes{\svara_1}{\svara_2})^-}\subseteq\valp{(\stimes{\svarb_2}{\svarb_1})^+}$, iff $(\valn{\svara_1^-}\times\valn{\svara_2^-})^{\bot}\subseteq\valp{\svarb_2^+}\times\valp{\svarb_1^+}$ by definition unfolding. Thus, we establish $\langle x,\stimes{y}{z}\rangle\in\bot$ on the assumptions (a)-(c). By (a), it suffices to show $\langle\stimes{y}{z},\stimes{u}{v}\rangle\in\bot$ given (d), (e), reducing to $\langle\sos{z}{u},\sos{v}{y}\rangle\in\bot$ by (F), iff $\langle\sos{v}{y},\sos{z}{u}\rangle\in\bot$. By (IH), (c) and (d), the desired result follows from $\sos{v}{y}\in(\valn{\svara_2^-}\leftarrow\valn{\svarb_2^-})^{\bot}$. But this is a consequence of (b), (e) and the fact that $\{\sos{v}{y} \ | \ v\in\valp{\svara_2^+}^{\bot}, y\in\valn{\svarb_2^-}^{\bot}\}\subseteq(\valn{\svara_2^-}\leftarrow\valn{\svarb_2^-})^{\bot}$ by L.\ref{galois}.

\end{proof}

\subsection{Completeness}
The purpose of this section is to demonstrate the completeness of strong focalization w.r.t. the phase models. Like in the previous section, we fix a set $X$ of formulas relative to which the relevant concepts are defined.

\begin{defs}
Completeness will be established w.r.t. the \textit{syntactic (phase) model}, defined by taking the structures $\Pi$ (relative to $X$) as phases, setting $\langle\Pi,\Sigma\rangle\in\bot$ iff $\cnlgseq{\Pi}{\Sigma}$ and letting $v(p)=\{ p\}^{\bot}=\{\Pi \ | \ \cnlgseq{\Pi}{p}\}$.
\end{defs}
\noindent The well-definedness of the syntactic model is a consequence of Lemmas \ref{closure_displ} and \ref{closure_grish}, the frame conditions for \calccnl \ being easily checked. The following is the central lemma of this section. 

\begin{lem}\label{truthlem_polar}
For arbitrary $N,P,\svard,\svare$, we have
\begin{center}
\begin{tabular}{ll}
(i) & $\svard\in\valn{N}$ implies $\cnlgseq{\svard}{\svare}$ for all $\svare\in\invp{N^{\bot}}$ \\
(ii) & $(\forall\svard)((\exists\svarf\in\invp{N^{\bot}})(\cnlgseqp{\svard}{\svarf}\Rightarrow\cnlgseq{\svard}{\svare})$ implies $\svare\in\valn{N}$ \\
(iii) & $\svard\in\valp{P}$ implies $\cnlgseq{\svard}{\svare}$ for all $\svare\in\invp{P}$ \\
(iv) & $(\forall\svard)((\exists\svarf\in\invp{P})(\cnlgseqp{\svard}{\svarf}\Rightarrow\cnlgseq{\svard}{\svare})$ implies $\svare\in\valp{P}$
\end{tabular}
\end{center}
\end{lem}
\begin{proof}
First, note that if for some $\svare\in\invp{N^{\bot}}$ ($\svare\in\invp{P}$) $\cnlgseq{\Pi}{\svare}$, then also $\cnlgseq{\svard}{N}$ ($\cnlgseq{\svard}{P^{\bot}}$) by applying $(D)$. Consequently, (ii) and (iv) imply, respectively, $N\in\uvaln{N}$ and $P^{\bot}\in\uvalp{P}$. In practice, when invoking the induction hypothesis for (ii) or (iv), we often immediately instantiate them by the latter consequences. To prove (i)-(iv), we proceed by a simultaneous induction on $P,N$. As typical cases, we check $p$, $\ntop N$, $P\os N$ and $P\lor Q$. \\

\noindent\textbf{Case} $p$. Since $\invp{p}=\{ p\}$, it suffices to show $\svard\in\valp{p}$ implies $\cnlgseq{\svard}{p}$ for (iii), and if $\cnlgseqp{\svard}{p}$ implies $\cnlgseq{\svard}{\svare}$ then also $\svare\in\valp{p}$, iff $\cnlgseq{\svare}{p}$ for (iv).
\begin{itemize}
\item[(iii)] By definition, as $\valp{p}=\{ p\}^{\bot}$.
\item[(iv)] Immediate from the observation that $\cnlgseqp{\svard}{p}$ iff $\svard=p$, as a simple case analysis on F.\ref{cnlg_strongfoc} will show.
\end{itemize}

\noindent\textbf{Case} $\ntop N$. Since $\invp{\ntop N}=\{ N\}$, it suffices to show $\svard\in\valp{\ntop N}$ implies $\cnlgseq{\svard}{N}$ for (iii), and if $\cnlgseqp{\svard}{N}$ implies $\cnlgseq{\svard}{\svare}$ then also $\svare\in\invp{\ntop N}$ for (iv).
\begin{itemize}
\item[(iii)] Suppose $\svard\in\valp{\ntop N}=\valn{N}^{\bot}$. By IV(ii), $N\in\valn{N}$, so that $\cnlgseq{\svard}{N}$.
\item[(iv)] We show $\svare\in\valp{\ntop N}=\valn{N}^{\bot}$, assuming (a) $\cnlgseqp{\svard}{N}$ implies $\cnlgseq{\svard}{\svare}$ for any $\svard$. Letting (b) $\svarf\in\valn{N}$, it suffices to ensure $\cnlgseq{\svare}{\svarf}$. IH(i) and (b) imply $\cnlgseq{\svarf}{\svarf'}$ for all $\svarf'\in\invp{N^{\bot}}$, hence $\cnlgseqp{\svarf}{N}$ by $(R)$. Thus, $\cnlgseq{\svarf}{\svare}$ by (a), and we apply $(dp)$.
\end{itemize}

\noindent\textbf{Case} $P\os N$. We show (iii) and (iv).
\begin{itemize}
\item[(iii)] Let (a) $\svard\in\valp{P\os N}$ and (b) $\svare\in\invp{P\os N}$, iff $\svare=\sos{\svare_1}{\svare_2}$ for some $\svare_1\in\invp{P}$ and $\svare_2\in\invp{N^{\bot}}$. We show $\cnlgseq{\svard}{\svare}$. By (a), it suffices to ensure $\svare_1\in\valp{P}^{\bot}$ and $\svare_2\in\valn{N}^{\bot}$. I.e., we must ascertain $\cnlgseq{\svare_1}{\svarf_1}$ and $\cnlgseq{\svare_2}{\svarf_2}$ on the assumptions $\svarf_1\in\valp{P}$ and $\svarf_2\in\valn{N}$. The desired result follows from IH(i), IH(iii), $(dp)$ and (b).
\item[(iv)] The following hypotheses will be used:
\begin{center}
\begin{tabular}{ll}
(a) & $\cnlgseqp{\svard}{\svarf}$ for some $\svarf\in\invp{P\os N}$ implies $\cnlgseq{\svard}{\svare}$ for all $\svard$ \\
(b) & $\svarf_1\in\valp{P}^{\bot}$ \\
(c) & $\svarf_2\in\valn{N}^{\bot}$ \\
(d) & $(\forall\svard_1)((\exists\svarf_1\in\invp{P})(\cnlgseqp{\svard_1}{\svarf_1})\Rightarrow\cnlgseq{\svard_1}{\stimes{\svarf_2}{\svare}})$ \\ & implies $\stimes{\svarf_2}{\svare}\in\valp{P}$ \\
(e) & $\cnlgseqp{\svard_1}{\svarf_1}$ for some $\svarf_1\in\invp{P}$ \\
(f) & $(\forall\svard_2)((\exists\svarf_2\in\invp{N^{\bot}})(\cnlgseqp{\svard_2}{\svarf_2})\Rightarrow\cnlgseq{\svard_1}{\sobs{\svare}{\svard_1}})$ \\ & implies $\sobs{\svare}{\svard_1}\in\valp{P}$ \\
(g) & $\cnlgseqp{\svard_2}{\svarf_2}$ for some $\svarf_2\in\invp{N^{\bot}}$
\end{tabular}
\end{center}
Assuming (a), we show $\svare\in\valp{P\os N}$. So let (b), (c). Since $\cnlgseq{\svare}{\sos{\svarf_1}{\svarf_2}}$ iff $\cnlgseq{\svarc_1}{\stimes{\svarf_2}{\svare}}$ by $(dp)$, it suffices by (b) to prove $\stimes{\svarf_2}{\svare}\in\valp{P}$. By (d), i.e., IH(iv), we need only prove $\cnlgseq{\svard_1}{\stimes{\svarf_2}{\svare}}$ on the assumption (e), iff $\cnlgseq{\svarf_2}{\sobs{\svare}{\svard_1}}$ by $(dp)$. Applying (c), we must show $\sobs{\svare}{\svard_1}\in\valn{N}$, which follows from (f), i.e., IH(ii), if we can prove $\cnlgseq{\svard_2}{\sobs{\svare}{\svard_1}}$ on the assumption (g). By $(dp)$ and (a), taking $\svarf=\sos{\svarf_1}{\svarf_2}$ in the latter case, this follows from $\cnlgseqp{\sos{\svard_1}{\svard_2}}{\sos{\svarf_1}{\svarf_2}}$, witnessed by $(\sos{}{})$, (e) and (f).
\end{itemize}

\noindent\textbf{Case} $P\lor Q$. We show (iii) and (iv).
\begin{itemize}
\item[(iii)] List of hypotheses:
\begin{center}
\begin{tabular}{ll}
(a) & $\svard\in\valp{P\lor Q}$, iff $\svard\in\valp{P}$ and $\svard\in\valp{Q}$ \\
(b) & $\svard\in\valp{P}$ implies $\cnlgseq{\svard}{\svare}$ for all $\svare\in\invp{P}$ \\
(c) & $\svard\in\valp{P}$ implies $\cnlgseq{\svard}{\svarf}$ for all $\svarf\in\invp{Q}$
\end{tabular}
\end{center}
Assume (a). The induction hypotheses (b) and (c) immediately imply $\cnlgseq{\svard}{\svard'}$ for all $\svard'\in\invp{P\lor Q}=\invp{P}\bigcup\invp{Q}$.
\item[(iv)] List of hypotheses:
\begin{center}
\begin{tabular}{ll}
(a) & $(\exists\svarf\in\invp{P\lor Q})(\cnlgseqp{\svard}{\svarf}$ implies $\cnlgseq{\svard}{\svare}$ for all $\svard$ \\
(b) & $(\forall\svard)((\exists\svarf'\in\invp{P})(\cnlgseqp{\svard}{\svarf'}\Rightarrow\cnlgseq{\svard}{\svare})$ implies $\svare\in\valp{P}$ \\
(c) & $\cnlgseqp{\svard}{\svarf'}$ for some $\svarf'\in\invp{P}$
\end{tabular}
\end{center}
We show $\svare\in\valp{P\lor Q}=\valp{P}\bigcap\valp{Q}$ on the assumption (a). We only prove $\svare\in\valp{P}$, with $\svare\in\valp{Q}$ following similarly. By the induction hypothesis (b), it suffices to show, for any given $\svard$, that $\cnlgseq{\svard}{\svare}$ on the assumption (c). Applying (a), we may suffice by demonstrating there exists $\svarf$ s.t. $\cnlgseqp{\svard}{\svarf}$. By (c), we may take $\svarf=\svarf'$.
\end{itemize}
\end{proof}
\begin{lem}\label{compl_sem_help} For arbitrary $N,P$, $\invp{N^{\bot}}\subseteq\valn{N}^{\bot}$, $\invp{P}\subseteq\valp{P}^{\bot}$
\end{lem}
\begin{proof}
We show the former statement, the second being dual. Suppose $\svare\in\invp{N^{\bot}}$. We apply L.\ref{truthlem_polar}(iv) to show $\svare\in\invp{\ntop N}(=\valn{N}^{\bot})$. So let $\svard$ be such that for some $\svarf\in\invp{\ntop N}$, $\cnlgseqp{\svard}{\svarf}$. Note, however, $\invp{\ntop N}=\{ N\}$, so that $\cnlgseqp{\svard}{N}$. Since the latter is derivable only by $(R)$, we have $\cnlgseq{\svard}{\svare'}$ for all $\svare'\in\invp{N^{\bot}}$, hence also $\cnlgseq{\svard}{\svare}$.
\end{proof}

\noindent We state completeness w.r.t. the syntactic model, implying in particular completeness w.r.t. all phase models.

\begin{thm}\label{completeness_sem} For arbitrary $N,P$, if $\valp{\ntop N}\subseteq\valp{P}$, then $\cnlgseq{\svard}{\svare}$ for every $\svard\in\invp{N^{\bot}}$ and $\svare\in\invp{P}$ with $X$ instantiated by $\{ N, P^{\bot}\}$.
\end{thm}
\begin{proof}
If $X=\{ N,P^{\bot}\}$, then $\invp{N^{\bot}}$ and $\invp{P}$ are well-defined relative to $X^{\tau}$ as argued in D.\ref{def_invp}, so that it makes sense to speak of presentations $\pres{\svard}{\svare}$ for every $\svard\in\invp{N^{\bot}}$ and $\svare\in\invp{P}$. Now, by L.\ref{compl_sem_help}, $\invp{N^{\bot}}\subseteq\valn{N}^{\bot}=\valp{\ntop N}$, hence $\invp{N^{\bot}}\subseteq\valp{P}$. The desired result follows immediately from L.\ref{truthlem_polar}(iii).
\end{proof}

\begin{cor}\label{final_main_theorem} We have $\cnlgseq{\svara}{\svarb}$ in \lgpol \ or \cnlpol \ iff $\cnlgseq{\svard}{\svare}$ for all $\svard\in\invp{\svara^+}$ and $\svare\in\invp{\svarb^+}$, instantiating $X$ by $\{\svara^-,\svarb^-\}$.
\end{cor}

\begin{proof}
The direction from left to right follows from composing Theorems \ref{soundness_sem} and \ref{completeness_sem}, while the composition of Theorems \ref{strong_unpol} and \ref{unpol_weak} takes care of the other direction.
\end{proof}

\section{Related Topics}
We consider some related topics and directions for future research.

\subsection{Synthetic inference rules}
Since the works of Girard (\cite{girard00}) and Andreoli (\cite{andreoli01}), the literature on focused proof search has become home to various implementations of synthetic inference rules, mostly concerning classical (linear) logic. The current account borrows a bit from everything, but is perhaps most similar to that of Zeilberger (\cite{zeilberger08}) in its depiction of the non-invertible phase, while more strongly resembling \cite{andreoli01} for the invertible rules. It should be noted, however, that Zeilberger's work stresses a higher order interpretation of focused proofs through the use of Martin-L\"ofs generalized inductive definitions, and proves normalization accordingly. 

\subsection{Focusing as a semantics of proofs}
Following Andreoli, we have explained focusing as a method of streamlining Cut-free backward chaining proof search. Around the same time as Andreoli's initial \cite{andreoli92}, however, Girard (\cite{girard91}) independently published on a similar sequent calculus (weakly focalized, by current terminology) for classical logic, with the aim of restoring the Church-Rosser property for Cut elimination, bypassing Lafont's critical pairs. In particular, Girard's results inspired a novel translation into intuitionistic logic, achieving parsimony by making the introduction of double negations contingent upon the polarity of the formula being translated. Focused derivations thus seem particularly suited to serve as a constructive theory of (classical) proofs, a theme further pursued by Zeilberger (\cite{zeilberger08}). 

The original intended application of \calclg \ and \calccnl \ being the study of natural language \textit{syntax} (argued to be similarly resource sensitive), an intuitionistic translation for focused derivations would be expected to similarly benefit investigations of natural language \textit{semantics} along the line of Montague's work (\cite{montague}). Double negation translations for \calclg \ have been previously studied by Bernardi and Moortgat (\cite{bernardimoortgat}) for precisely this purpose, although their work does not yet benefit of the structure of focused derivations.

\subsection{Normalization by evaluation}
Save for the na\"ive use of the set theoretic language, none of our proofs resort to classical reasoning. In particular, the completeness result of $\S$4.2 proceeds not via the usual proof by contraposition through the construction of countermodels, but rather shows directly that any `truth' in the syntactic model has a focused proof. Thus, through a formalization in a constructive meta language like Martin-L\"of type theory or the calculus of constructions, we might hope to explicate the underlying algorithmic content underlying our work, being a mapping of sequent derivations in \calclg \ or \calccnl \ into focused derivations. Already for intuitionistic logic, such formalizations of constructive completeness proofs have been studied by Coquand (\cite{coquand93}) and Herbelin and Lee (\cite{herbelinlee09}), while Ilik (\cite{ilik10}) additionally considers classical logic. Each of the works cited further stress the connection to \textit{normalization by evaluation}, first appearing in \cite{nbe_first}, seeking normalization proofs for the $\lambda$-calculus (and later, arbitrary term rewriting systems) making no recourse to the usual reduction relations. We leave the study of such connections for \calclg \ and \calccnl \ as future research.

\bibliography{bibliography}

\begin{thebibliography}{10}

\bibitem{abrusci02}
Vito~Michele Abrusci.
\newblock Classical conservative extensions of {L}ambek calculus.
\newblock {\em Studia Logica}, 71(3):277--314, 2002.

\bibitem{andreoli92}
Jean-Marc Andreoli.
\newblock Logic programming with focusing proofs in linear logic.
\newblock {\em Journal of Logic and Computation}, 2(3):297--347, 1992.

\bibitem{andreoli01}
Jean-Marc Andreoli.
\newblock Focussing and proof construction.
\newblock {\em Annals of Pure and Applied Logic}, 107(1-3):131--163, 2001.

\bibitem{andreoli04}
Jean-Marc Andreoli.
\newblock An axiomatic approach to structural rules for locative linear logic.
\newblock In {\em Linear logic in computer science}, volume 316 of {\em London
  mathematical society lecture notes series}. Cambridge University Press, 2004.

\bibitem{nbe_first}
Ulrich Berger and Helmut Schwichtenberg.
\newblock An inverse of the evaluation functional for typed
  $\lambda$--calculus.
\newblock In R.~Vemuri, editor, {\em Proceedings of the Sixth Annual IEEE
  Symposium on Logic in Computer Science}, pages 203--211. IEEE Computer
  Society Press, Los Alamitos, 1991.

\bibitem{bernardimoortgat}
R.~Bernardi and M.~Moortgat.
\newblock Continuation semantics for the {L}ambek-{G}rishin calculus.
\newblock {\em Information and Computation}, 208(5):397--416, 2010.

\bibitem{chaudhuripfenning}
K.~Chaudhuri and F.~Pfenning.
\newblock Focusing the inverse method for linear logic.
\newblock In {\em Computer Science Logic, 19th International Workshop, CSL
  2005}, volume 3634 of {\em Lecture Notes in Computer Science}, pages
  200--215. Springer, 2005.

\bibitem{cockettseely}
J.~R.~B. Cockett and R.~A.~G. Seely.
\newblock Weakly distributive categories.
\newblock In {\em Journal of Pure and Applied Algebra}, pages 45--65.
  University Press, 1991.

\bibitem{coquand93}
Catarina Coquand.
\newblock From semantics to rules: A machine assisted analysis.
\newblock In Egon B{\"o}rger, Yuri Gurevich, and Karl Meinke, editors, {\em
  Computer Science Logic, 7th Workshop, CSL '93, Swansea, United Kingdom,
  September 13-17, 1993, Selected Papers}, pages 91--105, 1993.

\bibitem{degrootelamarche}
{P}hilippe {D}e {G}roote and {F}ran{\c{c}}ois {L}amarche.
\newblock Classical non associative {L}ambek {C}alculus.
\newblock {\em {S}tudia {L}ogica}, 71:355--388, 2002.

\bibitem{girard91}
Jean-Yves Girard.
\newblock A new constructive logic: Classical logic.
\newblock {\em Mathematical Structures in Computer Science}, 1(3):255--296,
  1991.

\bibitem{girard00}
Jean-Yves Girard.
\newblock On the meaning of logical rules {II}: multiplicatives and additives.
\newblock In {\em Foundation of Secure Computation}, pages 183--212, Amsterdam,
  2000. IOS Press.

\bibitem{grishin}
V.N. Grishin.
\newblock On a generalization of the {A}jdukiewicz-{L}ambek system.
\newblock In A.~I. Mikhailov, editor, {\em Studies in Nonclassical Logics and
  Formal Systems}, pages 315--334, Nauka, Moscow, 1983.

\bibitem{herbelinlee09}
Hugo Herbelin and Gyesik Lee.
\newblock Forcing-based cut-elimination for {G}entzen-style intuitionistic
  sequent calculus.
\newblock In Hiroakira Ono, Makoto Kanazawa, and Ruy J. G.~B. de~Queiroz,
  editors, {\em Logic, Language, Information and Computation, 16th
  International Workshop, WoLLIC 2009, Tokyo, Japan, June 21-24, 2009.
  Proceedings}, pages 209--217, 2009.

\bibitem{ilik10}
Danko Ilik.
\newblock {\em Constructive completeness proofs and delimited control}.
\newblock PhD thesis, \'{E}cole Polytechnique, INRIA, Universit\'e Paris
  Diderot, 2010.

\bibitem{lamarche03}
{F}ran{\c{c}}ois {L}amarche.
\newblock On the algebra of structural contexts.
\newblock {\em Mathematical structures in computer science}, 2003.

\bibitem{lambek58}
Joachim Lambek.
\newblock The mathematics of sentence structure.
\newblock {\em American Mathematical Monthly}, 65:154--169, 1958.

\bibitem{lambek61}
Joachim Lambek.
\newblock On the calculus of syntactic types.
\newblock In Roman Jakobson, editor, {\em Structure of Language and its
  Mathematical Aspects, Proceedings of the Twelfth Symposium in Applied
  Mathematics}, 1961.

\bibitem{laurent04}
Olivier Laurent.
\newblock A proof of the focalization property of linear logic.
\newblock Unpublished note, May 2004.

\bibitem{montague}
R.~Montague.
\newblock Universal grammar.
\newblock {\em Theoria}, 36(3):373--398, 1970.

\bibitem{moortgat09}
Michael Moortgat.
\newblock Symmetric categorial grammar.
\newblock {\em Journal of Philosophical Logic}, 38(6):681--710, 2009.

\bibitem{moot07}
Richard Moot.
\newblock Proof nets for display logic.
\newblock {\em CoRR}, abs/0711.2444, 2007.

\bibitem{okada02}
Mitsuhiro Okada.
\newblock A uniform semantic proof for cut-elimination and completeness of
  various first and higher order logics.
\newblock {\em Theoretical Computer Science}, 281(1-2):471--498, 2002.

\bibitem{zeilberger08}
Noam Zeilberger.
\newblock On the unity of duality.
\newblock {\em Annals of Pure and Applied Logic}, 153(1-3):66--96, 2008.

\end{thebibliography}

\end{document}